\documentclass[lettersize,journal]{IEEEtran}
\usepackage{amsmath,amsfonts,amsthm,amssymb,bbm}
\usepackage{newtxmath} 
\usepackage{extarrows} 
\usepackage{algorithmicx}
\usepackage{algorithm}
\usepackage{array}
\usepackage[caption=false,font=footnotesize,labelfont=rm,textfont=rm]{subfig}
\usepackage{textcomp}
\usepackage{stfloats}
\usepackage{url}
\usepackage{verbatim}
\usepackage{graphicx}
\usepackage{cite}
\usepackage{hyperref} 
\hypersetup{
	hidelinks,
	colorlinks=true, 	
	linkcolor=black,
	citecolor=black
}

\usepackage{color}
\usepackage{color,xcolor}
\newtheorem{lemma}{\textbf{\textit{Lemma}}}
\newtheorem{theorem}{\textbf{\textit{Theorem}}}

\usepackage{titlesec}
\usepackage{indentfirst} 

\begin{document}
	
%
\title{Jamming Detection and Channel Estimation for Spatially Correlated Beamspace Massive MIMO}
%
%
%

\author{ Pengguang Du,~\IEEEmembership{Graduate Student Member,~IEEE,} Cheng Zhang,~\IEEEmembership{Member,~IEEE,}  Yindi Jing,~\IEEEmembership{Senior Member,~IEEE,}\\Chao Fang, Zhilei Zhang, Yongming Huang,~\IEEEmembership{Senior Member,~IEEE}
\thanks{Pengguang Du, Cheng Zhang, Chao Fang and Yongming Huang are with the National Mobile Communication Research Laboratory, and the School of Information Science and Engineering, Southeast University, Nanjing 210096, China, and also with the Purple Mountain Laboratories, Nanjing 211111, China (e-mail: pgdu@seu.edu.cn; zhangcheng\_seu@seu.edu.cn; fangchao@seu.edu.cn; huangym@seu.edu.cn). \emph{(Corresponding authors: Yongming Huang;	Cheng Zhang.)}}
\thanks{Yindi Jing is with the Department of Electrical and Computer Engineering, University of Alberta, Edmonton, AB T6G 1H9, Canada (e-mail: yindi@ualberta.ca).}
\thanks{Zhilei Zhang is with the Purple Mountain Laboratories, Nanjing 211111, China
(e-mail: zhangzhilei@pmlabs.com.cn).}
}

\markboth{}%
{Submitted paper}
\maketitle 

\titlespacing*{\subsection}{0pt}{0pt plus 1pt minus 1pt}{0.5pt plus 1pt minus 1pt}

\begin{abstract}
 In this paper, we investigate the problem of jamming detection and channel estimation during multi-user uplink beam training  under random pilot jamming attacks in beamspace massive multi-input-multi-output (MIMO) systems. 
 For jamming detection, we distinguish  the signals from the jammer and the user by projecting the observation signals onto the pilot space. By using the multiple projected observation vectors corresponding to the unused pilots, we propose a jamming detection scheme based on the locally most powerful test (LMPT) for systems with general channel conditions. Analytical expressions for the probability of detection and false alarms are derived using the second-order statistics and likelihood functions of the projected observation vectors.
 For the detected jammer along with users, we propose a two-step minimum mean square error (MMSE) channel estimation  using the projected observation vectors.
As a part of the channel estimation, we develop schemes to estimate the norm and the phase of the inner-product of the legitimate pilot vector and the random jamming pilot vector, which  can be obtained using linear MMSE estimation and a bilinear form of the multiple projected observation vectors.
From simulations under different system parameters, we observe that the proposed technique improves the detection probability by 32.22\% compared to the baseline at  medium channel correlation level, and the channel estimation achieves a mean square error of -15.93$\,\text{dB}$.

\end{abstract}


%
\IEEEpeerreviewmaketitle


\vspace{1em}
\section{Introduction}
\IEEEPARstart{B}{eamspace} massive multiple-input multiple-output (MIMO) based on circuit-type beamforming networks or quasi-optical lens can provide significant performance gains to enable future wireless communications \cite{wu2023simultaneous}.
However, the exposure of the physical channels in wireless communications makes it easy for non-cooperative jammers to launch active attacks to cause pilot contamination and degrade the spectral efficiency (SE) of the system \cite{hoang2018cell}.
Jamming detection and the channel state information (CSI) acquisition of the jammer and legitimate users are the top priorities in improving system security and guaranteeing its robust performance \cite{kapetanovic2013detection,tugnait2015self,vinogradova2016detection,akhlaghpasand2017jamming,qi2024anti}.
Jamming suppression techniques, e.g., anti-jamming oriented hybrid beamforming \cite{qi2024anti} and time-frequency resource allocation \cite{pirayesh2022jamming}, can be further employed to  mitigate jamming effects. 

Various jamming detection schemes, including uplink training-based one-way and uplink-downlink training-based two-way schemes, have been proposed for conventional antenna-domain massive MIMO systems.
The works on uplink training assume that the system uses different pilot transmission strategies, making it impossible for the jammer to know the user's pilot and only allowing the transmission of random sequences.
Specifically, the pilots are transmitted by superimposing a random sequence on a publicly known training sequence \cite{kapetanovic2013detection,tugnait2015self}.
The jammer is then detected using an inner-product-phase-based method and a minimum description length (MDL) method, respectively.
In \cite{vinogradova2016detection} and \cite{akhlaghpasand2017jamming}, a pilot hopping technique is used to prevent the jammer from estimating the user's pilots. An eigenvalue-based method and a generalized likelihood ratio test (GLRT) are then proposed to detect the jammer.
However, without changing the design of the pilot symbols and transmission protocols, a pilot spoofing attack can be launched, i.e., the jammer transmits the same pilot signals as those of the intended user. 
Here, the jammer disguises itself as a legitimate user, making it difficult for the  base station (BS) to detect the attack.
This problem has motivated research on uplink-downlink training   \cite{xiong2015energy,xiong2016secure,xu2019detection}. The training framework 
consists of the following phases:  the uplink pilot transmission, the downlink broadcasting of the modulated symbols, and the demodulation and detection by the user.
In the schemes proposed in \cite{xiong2015energy,xiong2016secure,xu2019detection}, if the downlink beam direction deviates from the user due to contaminated uplink CSI estimation,  the received signal strength (RSS) diminishes.
 Based on this observation, the following methods have been proposed:   an energy ratio detector \cite{xiong2015energy}, a two-way estimation detector \cite{xiong2016secure}, and a pilot retransmission detector \cite{xu2019detection}.

The reliable acquisition of the CSI of both  jammers and users is a fundamental prerequisite for the implementation of existing anti-jamming methods \cite{qi2024anti,zeng2017enabling}. Many researchers have developed schemes to estimate CSI under jamming conditions.
For example, a user channel estimator  using least squares or minimum-mean-square-error (MMSE) criteria was  proposed in \cite{gulgun2020massive}. 
Considering the stochastic nature of jamming pilots,  the authors of \cite{akhlaghpasand2019jamming} and \cite{do2017jamming} employed an unused pilot sequence to estimate the jamming channel and proposed an MMSE jamming suppression channel estimator based on  the asymptotic behavior of massive MIMO for legitimate users.
To maximize the system's SE under jamming attacks, an optimal linear estimator for user channels was devised in \cite{akhlaghpasand2020jamming}, leveraging statistical knowledge of the jammer's channel.
To improve estimation quality, a semi-blind estimation method was proposed in \cite{yang2022jamming}.  This method uses pilots and data symbols to jointly estimate the legitimate channel and the jamming-plus-noise (JPN) correlation matrix through an expectation-maximization algorithm.
By exploring angular domain knowledge, the work in \cite{bagherinejad2021direction} projected received signals onto the orthogonal complement of the jammer's angular subspace and estimated the user's path gains and CSI using the linear MMSE approch.

The aforementioned detection and estimation works mainly focus on conventional antenna-space MIMO systems, and the extension to beamspace systems using hybrid digital-analog hardware architectures is a significant challenge. 
This challenge arises because only limited compressed beam-domain information from the radio-frequency (RF) chain is available for physical-layer signal processing.
Similar to antenna-domain MIMO, the communication phase in beam-domain MIMO communications most vulnerable to jamming attacks is the pilot transmission, also known as beam training.
As described in \cite{kim2021adversarial,darsena2022anti}, a jammer can provoke a beam failure event by sending high-power RF perturbation to interfere with the RSS at the BS.
To address this problem, the papers \cite{darsena2022anti,hou2023music,dinh2023defensive} have proposed model-based and deep learning-driven jamming detection schemes.
In the model-based scheme, \cite{darsena2022anti} broadcasts training sequences through a randomized probing scheme. 
The power of the JPN is then computed using pilot projection to enable jamming detection.
The work in \cite{hou2023music} proposed a method for jamming detection and source localization based on the multiple signal classification spectrum.
By integrating deep learning,  an autoencoder-based jamming detection and mitigation defense strategy has been designed in \cite{dinh2023defensive}.
In regard to channel estimation, the aforementioned schemes enable the estimation of the beam-domain channel direction.
However, the limited beam-domain channel information hinders the development of refined beam designs, making it challenging to explore and enhance spatial diversity gain.

As previously mentioned, the antenna-domain CSI estimation for users and jammers in beamspace  massive MIMO systems has not yet been studied.
Additionally, most existing studies assume independent fading channels, overlooking the spatially correlated fading caused by small antenna spacing \cite{li2022spatial}.
To address these limitations, this paper proposes channel-statistics-assisted jamming detection scheme and MMSE-based two-step estimation methods for the antenna-domain channels of both the jammer and users for beamspace massive MIMO systems.
The key contributions  are outlined below.
\begin{itemize}
	\item We propose a multi-user uplink beam training model under random pilot jamming attack that contemplates beamspace massive MIMO systems under spatially correlated channel fading, and  utilizes projected observation vectors onto the used and unused pilot sequences to distinct the signals from legitimate users and the jammer.
	\item The hypothesis testing problem regarding the existence of jamming is formulated based on multiple projected observation vectors corresponding to the unused pilots. 
	We derive the likelihood functions corresponding to both the original and alternative hypotheses. 
	By treating the problem as a one-sided test, we propose a locally most powerful test (LMPT) based detection scheme for the general channel scenario.
	It utilizes the sum of the gradients of the projected observation vectors with respect to  the argument to be detected to enable jamming detection. By utilizing the full eigenspace information of the jammer channel, we obtain a tractable expression for the  statistic.	
	Further, we analyze its complexity and performance.
	 
	
	\item The unknown nature of the jamming pilot introduces unknown parameters in the channel estimation: the inner-products of the jamming and system pilots. 	
	Therefore, we propose an MMSE-based two-step channel estimation scheme. 
	First, we design a linear MMSE estimator for the norms of the inner-products, and an estimator for their phase differences  based on a bilinear form of the projected observation vectors.
	Then, we estimate the jamming and user channels  using the projected observation vectors and the estimated inner-products. 
	
	\item The simulation results validate our derived performance expressions for the LMPT-based jamming detector and the effectiveness of the proposed channel estimation method.
	Our proposed detector has better detection performance at high or medium channel correlation levels compared to the GLRT-based detector.
	Specifically, the detector improves the detection probability by 32.22\% 
	under the exponential correlation model when the channel correlation coefficient is 0.5.
	In addition, the performance of the proposed estimation method is comparatively analyzed under strong and weak channel correlation scenarios. The mean square error (MSE) of channel estimate  achieves -15.93$\,\text{dB}$ when the channel correlation coefficient is 0.8, with better performance at the channel correlation coefficient of 0.2.

\end{itemize}

The remainder of this paper is organized as follows. The system model and beam training model under jamming attack are presented in Section \ref{System Model}. The channel-statistics-assisted jamming detection scheme is detailed in Section \ref{Second-Order Statistics Based Jamming Detection Scheme}. This section includes the LMPT-based jamming detection scheme and its performance analysis. Section \ref{MMSE-based Spatial-Domain CE} presents  the two-step channel estimation approach for the jamming and user channels, and the estimation method for the jamming pilot's inner-product. Section \ref{Simulation} provides detailed simulations and discussions.  Finally, Section \ref{Conclusion} offers concluding remarks.

\emph{Notation}: $\mathbb{{C}}^{m \times n}$ indicates 
the set of $m \times n$ complex matrices. 
Bold symbols denote matrices and vectors.
$\mathbf{I}_{n}$ denotes the $n$-by-$n$ identity matrix. 
${\bf{x}} \sim \mathcal{{CN}}\left(\text{\ensuremath{\boldsymbol{\mu}}},\boldsymbol{\Sigma}\right)$ means that the random vector ${\bf{x}}$ follows the circularly symmetric complex Gaussian distribution with mean vector $\boldsymbol{\mu}$ and covariance matrix $\boldsymbol{\Sigma}$.
$\left( \cdot \right)^{\rm{T}}$, 
$\left( \cdot \right) ^{\rm{H}}$, $\left( \cdot \right) ^{*}$,
 ${\rm tr}\left( \cdot \right)  $ and $|\cdot|$ refer to the transpose,
conjugate transpose, complex conjugate, trace and determinant. 
Also, $\left\Vert \mathbf{a}\right\Vert $ and
$\left\Vert \mathbf{A}\right\Vert _{\rm{F}} $ denote the Euclidean norm of $\mathbf{a}$ and
the Frobenius norm of $\mathbf{A}$, respectively. 
$\rm{diag}(\mathbf{a})$ is the diagonal matrix whose diagonal entries are elements of vector $\mathbf{a}$.
$\mathbb{E}\left\{ \cdot\right\} $ 
is the expectation operator.  
The Kronecker product is denoted by the symbol $\otimes$.
Finally, $\chi _b^2\left( a \right)$ denotes a chi-squared distribution with $b$ degrees-of-freedom and the noncentrality parameter $a$.

\section{System Model and Beam Training Model Under Jamming Attack }
\label{System Model}

\par We consider a beamspace massive MIMO uplink system,  where the BS  with $M_{\rm{BS}}$ antennas serves $K$ legitimate users each with $M_{\rm{UE}}$ antennas via orthogonal frequency division multiple  access (OFDMA) \cite{zhang2023incremental}. 
A jammer equipped with $M_{\rm{JM}}$ antennas is within the coverage area, and attacks the uplink training of the  system by sending jamming signals.
The BS, users and the jammer are considered to use a single RF chain architecture with only analog beamforming capability.\footnote{In a single RF chain architecture, users can access the BS via frequency division multiplexing, and this can easily be extended to space division multiplexing scenarios in a hybrid multi-RF chain architecture.}

\vspace{0.5em}
\subsection{Pilot Transmission Model and Channel Model }

During the uplink training phase, the legitimate users send pilot sequences to the BS for channel estimation. 
Denote the set of legitimate pilots in the system as ${\boldsymbol{\Phi}}=\left\lbrace {\boldsymbol{\phi}}_{1},\ldots,{\boldsymbol{\phi}}_{\tau}\right\rbrace$ where ${\boldsymbol{\phi}}_{i} \in \mathbb{C}^{\tau \times 1}$, and $\tau$ is the pilot length.
 The pilot sequences are orthogonal to each other and each has the unit norm, i.e., $||{\boldsymbol{\phi}}_{i}||=1, \forall i\in \mathbb{T}= \{1,\ldots,\tau\}$.
Denote the set of users as $\mathbb{K}= \left\lbrace1,\ldots,K \right\rbrace $.  
For simplicity, it is assumed that user $k \in \mathbb{K}$ is assigned subband $k$.
User $k\in \mathbb{K}$ sends the pilot sequence with index ${\iota}_k \in \mathbb{T}$ in the $k$-th subband with a RF analog precoder ${\bf{w}}_{k} = \left[w_{k,1},\ldots,w_{k,M_{\rm{UE}}}\right]^{\rm{T}}\in \mathbb{C}^{M_{\rm{UE}} \times 1}$ satisfying the constant modulus constraint $|{{w}}_{k,i}| = \frac{1}{\sqrt{M_{\rm{UE}}}},i = 1,\ldots,{M_{\rm{UE}}}$, and the BS is assumed to use an analog combiner ${\bf{w}}_{{\rm{RF}}}\in\mathbb{C}^{M_{\rm{BS}} \times 1}$ to observe the received antenna-domain signal.
Without loss of generality, it is assumed that ${\bf{w}}_{{\rm{RF}}}$ is a unit-power vector.
Meanwhile, the jammer sends the jamming pilot sequence ${\boldsymbol{\psi}}_{k}$ in the $k$-th subband, where ${\boldsymbol{\psi}}_{k}$ follows $ \mathcal{CN}\left({\bf{0}},\frac{1}{\tau} {\bf{I}}_{\tau}\right) $.\footnote{\textcolor{black}{ For systems using pilot hopping where the jammer does not have the knowledge of the pilot sequences, the design criterion for jamming pilot is that the inner-product of ${\boldsymbol{\psi}}_{k}$ and each ${\boldsymbol{\phi}}_{i}$ in the pilot
		set is always non-zero. A random jamming pilot sequence satisfing above description divides its average power equally over all pilots, so this is a reasonable mechanism \cite{do2017jamming,akhlaghpasand2019jamming}.}}
Denote ${\bf{w}}_{\rm{JM}}\in \mathbb{C}^{M_{\rm{JM}}\times 1}$ as the analog beamforming vector of the jammer whose elements satisfy the constant modulus constraint.
Then, the received signal vector at the BS in the $k$-th subband, denoted as ${\bf{y}}_{k} \in \mathbb{C}^{\tau\times 1}$, can be written as the following:
\begin{align}
	\setlength{\abovedisplayskip}{3pt}
	\setlength{\belowdisplayskip}{3pt}
		{\bf{y}}_{k}^{\rm{T}}
		&= {\sqrt{\tau p_t}}{\bf{w}}_{{\rm{RF}}}^{\rm{H}}{\bf{H}}_{k}{\bf{w}}_{k}{\boldsymbol{\phi}}_{{\iota}_k }^{\rm{T}}+ {\sqrt{\tau q_k}}{\bf{w}}_{{\rm{RF}}}^{\rm{H}}{{\bf{H}}^{{\rm{JM}}}_{k}{\bf{w}}_{{\rm{JM}}}}{\boldsymbol{ \psi}}^{\rm{T}}_{k} 
		+{\bf{w}}_{{\rm{RF}}}^{\rm{H}}{\bf{N}}_{k},
		\vspace{-0.8em}
\end{align}
where $p_t$ and $q_k$ are the uplink pilot powers per
symbol of the legitimate users and the jammer, respectively.
${\bf{H}}_{k}\in\mathbb{C}^{ M_{\rm{BS}} \times M_{\rm{UE}}}$ is the uplink antenna-domain channel matrix from user $k$ to the BS, and ${\bf{H}}^{\rm{JM}}_{k}\in\mathbb{C}^{ M_{\rm{BS}} \times M_{\rm{JM}}}$ is the antenna-domain channel matrix from the jammer to the BS in the $k$-th subband.
${\bf{N}}_{k}\in\mathbb{C}^{ M_{\rm{BS}} \times \tau}$ denotes the uplink additive noise matrix received in the $k$-th subband by the BS. The elements of ${\bf{N}}_{k}$ are assumed to be independent and identically distributed (i.i.d.) following  $\mathcal{CN}\left(0,\sigma^2\right)$, and $\sigma^2$ denotes the noise variance.

The BS projects ${\bf{y}}_{k}$ onto the $i$-th pilot sequence to get
\begin{align} \label{general projection}
	{{{y}}}_{k,i}= {\boldsymbol{\phi}}^{\rm{T}}_{i }{\bf{y}}_{k}^{\rm{*}}&=  {\sqrt{\tau p_t}}1_{i={\iota}_k} {\bf{w}}^{\rm{H}}_{k}{\bf{H}}_{k}^{\rm{H}}{\bf{w}}_{{\rm{RF}}}
		+ {\sqrt{\tau q_k}}{\boldsymbol{\phi}}^{\rm{T}}_{i }{\boldsymbol{ \psi}}^{*}_{k}
		\nonumber
		\\&\hspace{2cm}\times{\bf{w}}^{\rm{H}}_{{\rm{JM}}}{\bf{H}}^{\rm{JM},H}_{k}{\bf{w}}_{{\rm{RF}}} +{\boldsymbol{\phi}}^{\rm{T}}_{i}{\bf{N}}^{\rm{H}}_{k}{\bf{w}}_{{\rm{RF}}},
\end{align}
where $1_{i={\iota}_k} $ is an indicator function that $1_{i={\iota}_k}  = 1$ if $i={\iota}_k$, and $1_{i={\iota}_k}  = 0$ if $i \ne {\iota}_k$.
Define 
\begin{align}
\alpha_{k,i} \triangleq {\boldsymbol{\phi}}^{\rm{T}}_{i}{\boldsymbol{ \psi}}^{*}_{k},
\end{align}
which is the inner-product of the $i$-th legitimate pilot vector and the random jamming pilot vector in the $k$-th subband. 
We define ${{\bf{h}}^{{\rm{JM}}}_{k}} = {{\bf{H}}^{\rm{JM}}_{k}}{\bf{w}}_{{\rm{JM}}}$, which is  the equivalent jammer channel.  Further define ${\bar{\bf{n}}}_{k} ={\bf{N}}^{\rm{H}}_{k}{\bf{w}}_{{\rm{RF}}}$, and $n_{k,i} ={\boldsymbol{\phi}}^{\rm{T}}_{i}{\bar{\bf{n}}}_k$ which is the equivalent noise.
The projected signal can be rewritten as
\begin{align} \label{general projection2}
	\vspace{-0.4em}
	\setlength{\abovedisplayskip}{2pt}
	\setlength{\belowdisplayskip}{10pt}
	&{{{y}}}_{k,i}= {\sqrt{\tau p_t}}1_{i={\iota}_k} \left ({{\bf{w}}}^{\rm{H}}_{k}\otimes{{\bf{w}}}^{\rm{T}}_{{\rm{RF}}}\right) {{\bf{h}}}^{*}_{k} 
	+ {\sqrt{\tau q_k}}\alpha_{k,i}{{\bf{w}}}^{\rm{T}}_{{\rm{RF}}} {{\bf{h}}^{{\rm{JM}},*}_{k}}+n_{k,i}, 
	\vspace{-0.4em}
\end{align}
where ${{\bf{h}}}_{k} = {\rm{vec}}\left( {{\bf{H}}}_{k}\right) $ is the vectorized form of the channel. It can be shown that  ${\bar{\bf{n}}}_{k} \sim \mathcal{CN}\left(0,\sigma^2 {\bf{I}}_{\tau}\right)$ and $n_{k,i} \sim \mathcal{CN}\left(0,\sigma^2\right)$.
For the equivalent noise projected onto the different pilots, we have $ \mathbb{E}\left[n_{k,i} n^{*}_{k,j}  \right] = {\boldsymbol{\phi}}^{\rm{T}}_{i}\mathbb{E}\left[{\bar{\bf{n} }}_{k} {\bar{\bf{n} }}^{\rm{H}}_{k}  \right] {\boldsymbol{\phi}}^{*}_{j} = 0$, $i\ne j$, and thus $n_{k,i}$ is independent for different $i$.
A tractable way to model spatially correlated channels is the correlated Rayleigh fading model:  
${{\bf{h}}}_{k}\sim\mathcal{CN}\left( {\bf{0}},{\bf{R}}_k \right) $, where ${\bf{R}}_{k} $ is the channel covariance matrix. 
For the equivalent jammer channel ${{\bf{h}}}^{{\rm{JM}}}_{k}$,
we denote its channel covariance matrix as  ${\bf{R}}_{{\rm{JM}},k}$, then
${{\bf{h}}}^{{\rm{JM}}}_{k}\sim\mathcal{CN}\left( {\bf{0}},{\bf{R}}_{{\rm{JM}},k} \right)$ \cite{sanguinetti2019toward}.


\begin{figure}[!t]
	\centering
	\includegraphics[scale=0.44]{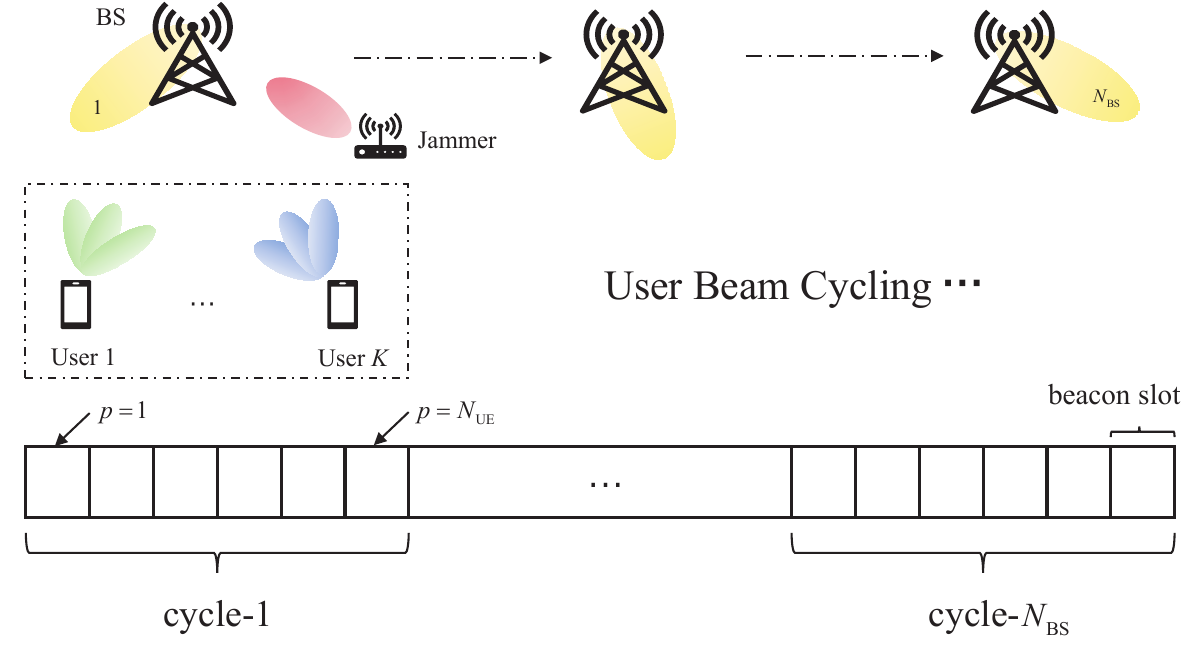}
	\vspace{-0.6cm}
	\caption{Multi-user beam cycle training procedure under jamming attacks.}
	\label{fig0}
\end{figure}

\vspace{0.5em}
\subsection{Beam Training Model}

As mentioned above, the BS's combiner ${\bf{w}}_{{\rm{RF}}}$ and the user's precoder ${\bf{w}}_{k}$ are used over one beacon slot, which has length $\tau$ for the transmission of one pilot sequence.
Further, the analog beam codebooks for the BS and users are defined as $\mathcal{F}_{\rm{BS}} = \left\lbrace {\bf{w}}_{{\rm{RF}},1},\ldots,{\bf{w}}_{{\rm{RF}},N_{\rm{BS}}}\right\rbrace $ and $\mathcal{F}_k =\left\lbrace {\bf{w}}_{k,1},\ldots, {\bf{w}}_{k,N_{\rm{UE}}}\right\rbrace$,  where $N_{\rm{BS}}$ and $N_{\rm{UE}}$ denote the sizes of codebooks respectively. 
Within an arbitrary beacon slot, the BS and users select a set of beam pairs in the pre-designed codebook for pilot transmission.
It is assumed that the precoder of the jammer remains fixed during the training.
As illustrated in Fig. \ref{fig0}, we adopt a beam cycling scheme \cite{lim2020efficient} with $N_{\rm{BS}}$ cycles, where in each cycle, the user sequentially transmits $N_{\rm{UE}}$ beams in the codebook. The number of beam slots of the training phase is thus $N_{\rm{b}} = N_{\rm{BS}} N_{\rm{UE}}$.


Specifically, within  cycle $q\in \left\lbrace 1,\ldots,N_{\rm{BS}}\right\rbrace $, user $k$ sends length-$\tau$ pilot signals over $N_{\rm{UE}}$ beacon slots through the beams in codebook $\mathcal{F}_k$, and the BS  receives signals with the beam ${\bf{w}}_{{\rm{RF}},q}$.
Denote the vector of the projected signals over the pilot sequence ${\boldsymbol{\phi}}_{i}$  in the $k$-th subband collected at the BS for the $q$-th beam cycle as  ${\bar{\bf{y}}}_{k,i,q}\in \mathbb{C}^{N_{\rm{UE}}\times 1}$. From Eq. \eqref{general projection2}, ${\bar{\bf{y}}}_{k,i,q}$
can be expressed as
\begin{equation}\begin{aligned} \label{projection q-cycle}
		\setlength{\abovedisplayskip}{3pt}
		\setlength{\belowdisplayskip}{3pt}
		{\bar{\bf{y}}}_{k,i,q}
		&= {\sqrt{\tau p_t}}1_{i={\iota}_k}{\bf{U}}_{k,q}{{\bf{h}}}^{*}_{k} + {\sqrt{\tau q_k}}\alpha_{k,i}{{\bf{U}}}_{{\rm{JM}},q} {{\bf{h}}}^{{\rm{JM}},*}_{k}
			+{\bf{n}}_{k,i,q},
\end{aligned}\end{equation}
where ${\bf{n}}_{k,i,q} = {\boldsymbol{\phi}}^{\rm{T}}_{i}{\bf{N}}^{\rm{H}}_{k}{\bf{w}}_{{\rm{RF}},q}$ is the noise vector whose elements are independent of each other. 
Define ${\tilde{\bf{w}}}_{k,q,p} =  {{\bf{w}}}^{\rm{H}}_{k,p}\otimes{{\bf{w}}}^{\rm{T}}_{{\rm{RF}},q} $ which is the combined transceiver beamforming vector corresponding to the beam pair $\left({\bf{w}}_{{\rm{RF}},q},{\bf{w}}_{k,p} \right) $, where $p$ is the index number of user's beam. Then
  ${\bf{U}}_{k,q}$ $=$ $\left[ {\tilde{\bf{w}}}^{\rm{H}}_{k,q,1},\ldots,{\tilde{\bf{w}}}^{\rm{H}}_{k,q,N_{\rm{UE}}}\right]^{\rm{H}} $, and ${{\bf{U}}}_{{\rm{JM}},q} =  \left[ {{\bf{w}}}_{{\rm{RF}},q} ,\ldots,{{\bf{w}}}_{{\rm{RF}},q}\right] ^{\rm{T}} \in \mathbb{C }^{N_{\rm{UE}} \times M_{\rm{BS}}}$.



By stacking projected observation vectors ${\bar{\bf{y}}}_{k,i,q}$ for $q=1,\ldots,N_{\rm{BS}}$, the received projected signal vector for the whole training phase can be presented as 
${\bar{\bf{y}}}_{k,i}= [{\bar{\bf{y}}}^{\rm{T}}_{k,i,1},\ldots,{\bar{\bf{y}}}^{\rm{T}}_{k,i,N_{\rm{BS}}}]^{\rm{T}} \in \mathbb{C}^{N_{\rm{b}}\times 1} $.
 From Eq. \eqref{projection q-cycle}, we have
\begin{equation}\begin{aligned} \label{high-dim signal vector}
		\setlength{\abovedisplayskip}{1pt}
		\setlength{\belowdisplayskip}{1pt}
		{\bar{\bf{y}}}_{k,i}
		&= 	{\sqrt{\tau p_t}}1_{i={\iota}_k}{\bf{U}}_{k}^{\rm{H}}{{\bf{h}}}^{*}_{k} + {\sqrt{\tau q_k}}\alpha_{k,i}{{\bf{U}}}_{\rm{JM}} ^{\rm{H}}{{\bf{h}}}^{{\rm{JM}},*}_{k}
			+{\bf{n}}_{k,i},
\end{aligned}\end{equation}
where ${\bf{U}}_{k} = [{\bf{U}}^{\rm{H}}_{k,1},\ldots,{\bf{U}}^{\rm{H}}_{k,N_{\rm{BS}}}] \in\mathbb{C}^{M_{\rm{BS}}M_{\rm{UE}} \times N_{\rm{b}}}$  represents the
combined transceiver beamforming matrix, ${\bf{U}}_{\rm{JM}} = [{\bf{U}}^{\rm{H}}_{{\rm{JM}},1},\ldots,{\bf{U}}^{\rm{H}}_{{\rm{JM}},N_{\rm{BS}}}] \in\mathbb{C}^{M_{\rm{BS}}\times N_{\rm{b}}}$, and ${\bf{n}}_{k,i} = [{\bf{n}}^{\rm{T}}_{k,i,1},\ldots,{\bf{n}}^{\rm{T}}_{k,i,N_{\rm{BS}}}]^{\rm{T}}$.
Since multiple received beams are realized at different beacon slots, ${\bf{n}}_{k,i,q}$ is mutually independent for different $q$, and thus we have ${\bf{n}}_{k,i} \sim\mathcal{CN}(0,\sigma^2{\bf{I}}_{N_{\rm{b}}})$.

Without loss of generality, let $\iota_k = 1$, then 
the index of the used pilot is 1 and the index set of the unused pilots is $\vmathbb{t} = \left\lbrace 2,\ldots,\tau\right\rbrace $. From Eq. \eqref{high-dim signal vector}, the projected observation vector corresponding to the used pilot received by the BS is given by
\begin{equation}\begin{aligned} \label{projected signal used}
		{\bar{\bf{y}}}_{k,1 }
		&=  {\sqrt{\tau p_t}}{\bf{U}}_{k}^{\rm{H}}{{\bf{h}}}^{*}_{k} + {\sqrt{\tau q_k}}\alpha_{k,1}{{\bf{U}}}^{\rm{H}}_{\rm{JM}}{{\bf{h}}}^{{\rm{JM}},*}_{k}
		+{\bf{n}}_{k,1},\\
\end{aligned}\end{equation}
and the projected observation vector corresponding to an unused pilot with index $i$ is given by
\begin{equation}\begin{aligned} \label{projected signal}
		{\bar{\bf{y}}}_{k,i}
		&=   {\sqrt{\tau q_k}}\alpha_{k,i}{{\bf{U}}}^{\rm{H}}_{\rm{JM}} {{\bf{h}}}^{{\rm{JM}},*}_{k}
		+{\bf{n}}_{k,i}, i \in \vmathbb{t}. \\
\end{aligned}\end{equation}
In the next section, we investigate the problem of jamming detection by jointly using the multiple projected observation vectors onto unused pilots in Eq. \eqref{projected signal}.

\section{Channel-Statistics-Assisted Jamming Detection Scheme}\label{Second-Order Statistics Based Jamming Detection Scheme}

In this section, we first formulate the hypothesis testing problem for jamming detection based on the multiple projected observation vectors. Then,  a channel-statistics-assisted jamming detection scheme is proposed based on the LMPT. Finally, the performance of the proposed LMPT-based jamming detection scheme is analyzed.

\vspace{0.5em}
\subsection{Formulation of the Hypothesis Testing Problem}

We formulate the jamming detection problem based on projected observation vectors corresponding to the unused pilots
$\left\lbrace  {\bar{\bf{y}}}_{k,i}, i\in \vmathbb{t}\right\rbrace $ in Eq. \eqref{projected signal} since they 
 contain the jammer signals but not user signals.
 The hypothesis test is defined as follows, $\forall k$,
\begin{equation}\begin{aligned} \label{Hypothesis testing}
		&\mathcal{H}_0 :{\bar{\bf{y}}}_{k,i } = {\bf{n}}_{k,i},& i \in \vmathbb{t},\\
		&\mathcal{H}_1 :{\bar{\bf{y}}}_{k,i } = {\sqrt{\tau q_k}}\alpha_{k,i}{{\bf{U}}}^{\rm{H}}_{\rm{JM}} {{\bf{h}}}^{{\rm{JM}},*}_{k}
		+{\bf{n}}_{k,i},&  i \in \vmathbb{t}.\\
\end{aligned}\end{equation}
The hypotheses $\mathcal{H}_0$ and $\mathcal{H}_1$ denote the cases where the jammer is absent and present, respectively.
We construct the vector containing all observed samples as ${{\bar{\bf{y}}}_{k}  } = \left [{{\bar{\bf{y}}}^{\rm{T}}_{k,2 }  },\ldots,{{\bar{\bf{y}}}^{\rm{T}}_{k,\tau }  }  \right ]^{\rm{T}} \in \mathbb{C}^{(\tau-1)N_{\rm{b}}\times 1}$, and denote the covariance matrix of ${{\bar{\bf{y}}}_{k}  }$ as ${{\bf{R}}}_{{\bar{\bf{y}}}_{k }}  = \mathbb{E}\left[ \bar{\mathbf{y}}_k\bar{\mathbf{y}}_k^{\rm{H}}\right] $.
Further define ${\tilde{\bf{h}}}^{{\mathrm{JM}}}_{k} \triangleq {\mathbf{U}}_{{\mathrm{JM}}}^\mathrm{H}\mathbf{h}_{k}^{{\mathrm{JM}},*}$ which is the beam-domain channel of the jammer, and  it can be shown that ${\tilde{\bf{h}}}^{{\mathrm{JM}}}_{k} \sim\mathcal{CN}\left( {\bf{0}} , {\tilde{\bf{R}}}_{{\mathrm{JM}},k} \right) $ where
${\tilde{\bf{R}}}_{{\rm{JM}},k} = {{\bf{U}}}^{\rm{H}}_{\rm{JM}} {{\bf{R}}}^{*}_{{\rm{JM}},k}{{\bf{U}}}_{\rm{JM}}$.
Under $\mathcal{H}_1$, the $i$-th $N_{\rm{b}}\times N_{\rm{b}}$ diagonal block matrix of ${{\bf{R}}}_{{\bar{\bf{y}}}_{k }}$ which is the self-covariance matrix of ${\bar{\bf{y}}}_{k,i }$ can be calculated as the following
\begin{equation}\begin{aligned} \label{self-cov}
		{{\bf{R}}}_{{\bar{\bf{y}}}_{k,i }} &=  \mathbb{E}\left[{{\bar{\bf{y}}}_{k,i }  }^{}{{\bar{\bf{y}}}^{\rm{H}}_{k,i }  }\right] \overset{(a)}{=} 
		{{\tau q_k}} |\alpha_{k,i}|^2 
		\tilde{\mathbf{R}}_{{\rm{JM}},k} + \sigma^2{\bf{I}}_{N_{\rm{b}}},i \in \vmathbb{t},\\
\end{aligned}\end{equation} 
where $(a)$ exploits the property that the jamming channel and the noise are independent of each other.

Denote the $(i, j)$-th $(i \ne j)$ $N_{\rm{b}}\times N_{\rm{b}}$ block matrix of ${{\bf{R}}}_{{\bar{\bf{y}}}_{k }}$ as $	{{\bf{R}}}_{{\bar{\bf{y}}}_{k,i },{\bar{\bf{y}}}_{k,j}}$, which is the mutual covariance matrix of the projected observation vectors ${\bar{\bf{y}}}_{k,i }$ and ${\bar{\bf{y}}}_{k,j }$, i.e., ${{\bf{R}}}_{{\bar{\bf{y}}}_{k,i },{\bar{\bf{y}}}_{k,j}} = \mathbb{E}\left[{{\bar{\bf{y}}}_{k,i }  }^{}{{\bar{\bf{y}}}^{\rm{H}}_{k,j}  }\right]$. According to Eq. \eqref{projected signal}, it can be obtained that
\begin{equation}\begin{aligned} \label{corss-cov} 
		{{\bf{R}}}_{{\bar{\bf{y}}}_{k,i },{\bar{\bf{y}}}_{k,j}}
		& =  {{\tau q_k}} {\alpha}_{k,i}{\alpha}^{*}_{k,j}\tilde{\mathbf{R}}_{{\rm{JM}},k} + \mathbb{E}\left[{\bf{n}}_{k,i}{\bf{n}}^{\rm{H}}_{k,j}\right],i \ne j.\\
\end{aligned}\end{equation}
Referring to the discussion of the equivalent noise projected onto different pilots after Eq. \eqref{general projection2}, we obtain $\mathbb{E}\left[{\bf{n}}_{k,i}{\bf{n}}^{\rm{H}}_{k,j}\right] = {\bf{0}}$, $\forall i\ne j$.
Further define  
\begin{align}
	{\boldsymbol{\alpha}}_{k} = \left[{\alpha}_{k,2},\ldots, {\alpha}_{k,\tau}\right]^{\rm{T}} ,\nonumber
\end{align}
which represents the inner-product vector of the legitimate pilots and the random jamming pilot, referred to as the jamming pilot's inner-product vector.
By combining Eqs. \eqref{self-cov} and \eqref{corss-cov}, we have
\begin{equation}
	\begin{aligned} \label{all-cov}
		{{\bf{R}}}_{{\bar{\bf{y}}}_{k }}&=\tau q_k\left( \boldsymbol{\alpha}_k\boldsymbol{\alpha}_k^{\rm{H}}\right) \otimes\tilde{\mathbf{R}}_{{\rm{JM}},k}+\sigma^2{\bf{I}}_{\left( \tau - 1\right)  N_{\rm{b}}}.
	\end{aligned}
\end{equation}

\newcounter{TempEqCnt} 
\setcounter{TempEqCnt}{\value{equation}} 
\setcounter{equation}{19} 
\begin{figure*}[hb]
	\centering 
	\hrulefill 
	\vspace*{-2pt} 
	\begin{align} \label{log likelihood}
		\ln{f\left( {\bar{\bf{y}}}_{k };\left\lbrace \bar{q}_{k,i} \right\rbrace_{i=2}^{\tau},\mathcal{H}_1\right) } &\approx \sum_{i=2}^{\tau} \ln{f\left( {\bar{\bf{y}}}_{k,i };\bar{q}_{k,i},\mathcal{H}_1\right) },\nonumber \\
		&	=-{{\left( \tau-1\right) N_{\rm{b}}}}\ln \pi   {- \sum_{i=2}^{\tau}\sum_{n=1}^{N_{\rm{b}}}\ln \left( \bar{q}_{k,i} \lambda_{k,n} +\sigma^2\right) }
			{- \sum_{i=2}^{\tau}{\bar{\bf{y}}}_{k,i }^{\rm{H}}{\bf{V}}_k\left( \bar{q}_{k,i}{\boldsymbol{\Lambda}}_k+ \sigma^2{\bf{I}}_{N_{\rm{b}}}\right) ^{-1}{\bf{V}}^{\rm{H}}_k{\bar{\bf{y}}}_{k,i }},\nonumber \\
			&=-{{\left( \tau-1\right) N_{\rm{b}}}}\ln \pi   {-\sum_{i=2}^{\tau}\sum_{n=1}^{N_{\rm{b}}}\ln \left( \bar{q}_{k,i}\lambda_{k,n} +\sigma^2\right) }
			{-\sum_{i=2}^{\tau}\sum_{n=1}^{N_{\rm{b}}}\frac{|{\bar{\bf{y}}}^{\rm{H}}_{k,i}  {{\bf{v}}}_{k,n} |^2}{\bar{q}_{k,i}\lambda_{k,n}+\sigma^2}   }.
	\end{align}
\vspace*{-15pt} 
\end{figure*}
\setcounter{equation}{\value{TempEqCnt}} 

Next, we analyze the
likelihood functions of the vector ${{\bar{\bf{y}}}_{k}  }$. 
It has been shown previously that the beam-domain channel ${\tilde{\bf{h}}}^{{\mathrm{JM}}}_{k}$  follows the Gaussian distribution $\mathcal{CN}\left( {\bf{0}} , {\tilde{\bf{R}}}_{{\mathrm{JM}},k} \right) $.
Within the uplink training interval, we regard the jamming pilot's inner-product vector ${\boldsymbol{\alpha}}_{k}$ to be deterministic, and thus the projected observation vectors $\left\lbrace {\bar{\bf{y}}}_{k,i },i\in\vmathbb{t}\right\rbrace $ under the hypothesis of $\mathcal{H}_1$ are joint complex Gaussian vectors, equivalently, ${\bar{\bf{y}}}_{k }$ is a complex Gaussian vector.
In terms of statistical significance, the $(i,j)$-th term of $\mathbb{E}\left[ \boldsymbol{\alpha}_k\boldsymbol{\alpha}_k^{\rm{H}}\right] $ is $\mathbb{E}\left[{\alpha}_{k,i}{\alpha}^{*}_{k,j}\right]  = \frac{1}{\tau} {\boldsymbol{\phi}}^{\rm{T}}_{i} {\boldsymbol{\phi}}^{*}_{j}$, and thus the cross-term and variance term are 0 and $\frac{1}{\tau}$, respectively.
As an approximation, we ignore the correlation between ${\alpha}_{k,i}$ and ${\alpha}_{k,j},\forall i \ne j$, and treat $\left\lbrace {\bar{\bf{y}}}_{k,i },i\in\vmathbb{t}\right\rbrace $ as uncorrelated and thus independent.
Under $\mathcal{H}_1$ and a given effective subband jamming power $\bar{q}_{k,i} = {{\tau q_k}} |\alpha_{k,i}|^2 $, the likelihood function of ${\bar{\bf{y}}}_{k}$, which is the joint probability density function (PDF) of  
${\bar{\bf{y}}}_{k,i},i\in \vmathbb{t}$, is approximated as the following
\begin{equation} \label{joint pdf}
	f\left({\bar{\bf{y}}}_{k };\left\lbrace \bar{q}_{k,i} \right\rbrace_{i=2}^{\tau}, \mathcal{H}_1\right)  \approx \prod_{i=2}^{\tau} {f\left( {\bar{\bf{y}}}_{k,i };\bar{q}_{k,i},\mathcal{H}_1\right) },
\end{equation}
where the PDF for the $N_{\rm{b}}$-dimensional vector ${\bar{\bf{y}}}_{k,i},i\in\vmathbb{t}$ is
\begin{equation} 
		\begin{aligned} \label{pdf h1}
	f\left( {\bar{\bf{y}}}_{k,i };\bar{q}_{k,i},\mathcal{H}_1\right)  =  \frac{\exp\left[-{\bar{\bf{y}}}_{k,i }^{\rm{H}}\left( {{\bar{q}_{k,i} }}{\tilde{\bf{R}}}_{{\rm{JM}},k} + \sigma^2{\bf{I}}_{N_{\rm{b}}}\right) ^{-1}{\bar{\bf{y}}}_{k,i }\right]}{\pi^{{N_{\rm{b}}}}\det^{}\left( {{{\bar{q}_{k,i}} }}{\tilde{\bf{R}}}_{{\rm{JM}},k} + \sigma^2{\bf{I}}_{N_{\rm{b}}}\right) } .
\end{aligned}\end{equation}
By setting \(\bar{q}_{k,i} = 0\) in Eqs. \eqref{joint pdf} and \eqref{pdf h1}, the likelihood function under $\mathcal{H}_0$ can be obtained as
\begin{equation} \label{h0 likelihood}
	f\left( {\bar{\bf{y}}}_{k};\mathcal{H}_0\right)  =    {\left( \pi\sigma^2\right) ^{-\left( \tau-1\right) {N_{\rm{b}}}}}   
	{\exp\left( -\frac{1}{\sigma^2}\|{\bar{\bf{y}}}_{k}\|^2\right) }.
\end{equation}
We can see from Eq. \eqref{pdf h1} that different to jamming detection under conventional independent Rayleigh fading channels  in massive MIMO \cite{akhlaghpasand2017jamming}, the likelihood function in the hypothesis test for jamming detection under spatially correlated beamspace channels depends on the channel statistics of the jammer.
This enables us to study channel-statistics-assisted jamming detection scheme.

\vspace{0.5em}
\subsection{LMPT-Based Detection Assisted by Channel Statistics}\label{LMPT-Based Jamming detection scheme}

For the hypothesis testing problem in Eq. \eqref{Hypothesis testing}, the effective jamming  power ${\bar{q}_{k,i}}$ is nonnegative, making it a one-sided test.
The uniform most powerful test (UMPT) offers the best detection performance for this hypothesis test, but it rarely exists when the problem  contains unknown parameters \cite{wang2018detection}.
In the absence of the UMPT, the LMPT is widely used 
as an asymptotic form of the UMPT for the one-sided and close hypothesis testing problem \cite{mohammadi2022generalized}.

The LMPT statistic can be derived by using a first-order Taylor expansion of the log-likelihood-ratio at 
\(\bar{q}_{k,i} = 0,i\in \vmathbb{t}\). By eliminating the high-order infinitesimals, the LMPT statistic is
\begin{align}
		&\quad \sum_{i=2}^{\tau} \ln{f\left( {\bar{\bf{y}}}_{k,i };\bar{q}_{k,i},\mathcal{H}_1\right) }-\sum_{i=2}^{\tau}\ln{f\left( {\bar{\bf{y}}}_{k,i };\mathcal{H}_0\right) } \nonumber
		\\ &\approx  \sum_{i=2}^{\tau}\bar{q}_{k,i}\frac{\partial \ln f\left( {\bar{\bf{y}}}_{k,i };\bar{q}_{k,i},\mathcal{H}_1\right) }{\partial\bar{q}_{k,i}}\bigg|_{\bar{q}_{k,i} = 0}.
\end{align}
%
To deal with the unknown effective jamming power $\bar{q}_{k,i}$,
we replace it with its statistical mean: 
$\mathbb{E}\left [ \bar{q}_{k,i} \right ] = {{\tau q_k}} \mathbb{E}\left[ |\alpha_{k,i}|^2 \right] = {{\tau q_k}} {\boldsymbol{\phi}}^{\rm{T}}_{i}
\mathbb{E}\left [ {\boldsymbol{ \psi}}^{*}_{k} {\boldsymbol{ \psi}}^{\rm{T}}_{k}
\right ]  {\boldsymbol{\phi}}^{\rm{*}}_{i} =   q_k $.
This can be viewed as an approximation when the pilot length $\tau$ is large.
Thus, the LMPT-based detector can be transformed into
\begin{equation} \label{lmpt-1}
	\sum_{i=2}^{\tau}\frac{\partial \ln f\left( {\bar{\bf{y}}}_{k,i };{q}_k,\mathcal{H}_1\right) }{\partial {q}_k}\bigg|_{{q}_k = 0}
	\stackrel{\mathcal{H}_{1}}{\underset{\mathcal{H}_{0}}{\gtrless}}  \frac{\ln\gamma_{\rm{LMPT}}}{{q}_k}.
\end{equation}
One interpretation of the detector is  that the sum of the gradients at ${q}_k = 0$ of the PDFs of the projected observation vectors under $\mathcal{H}_1$ will exceed the decision threshold if the jammer is present.
To achieve a tractable form for detection, we need to further simplify Eq. \eqref{pdf h1}.
We denote the rank of the beam-domain channel covariance matrix of the jammer ${\tilde{\bf{R}}}_{{\rm{JM}},k}$ as $\rho_k$ and consider the eigenvalue decomposition (EVD) of ${\tilde{\bf{R}}}_{{\rm{JM}},k}$:
$
{\tilde{\bf{R}}}_{{\rm{JM}},k} = {\bf{V}}_k{\boldsymbol{\Lambda}}_k{\bf{V}}^{\rm{H}}_k
$,
where ${\boldsymbol{\Lambda}}_k = {\rm{diag}}\left \{ \left [ \lambda_{k,1},\ldots,\lambda_{k,{N_{\rm{b}}}} \right ]  \right \} $ with $\lambda_{k,1}\ge \lambda_{k,2}\ge \ldots \ge\lambda_{k,\rho_k} >0$ and $\lambda_{k,\rho_k+1}= \ldots =\lambda_{k,{N_{\rm{b}}}} =0$, and ${\bf{V}}_k = \left [ {\bf{v}}_{k,1},\ldots,{\bf{v}}_{k,{N_{\rm{b}}}}  \right ] $ is the eigenvector matrix. Therefore,
\begin{align} \label{det}
	\det\left( \bar{q}_{k,i}{\tilde{\bf{R}}}_{{\rm{JM}},k} + \sigma^2{\bf{I}}_{N_{\rm{b}}}\right)  
	&= \prod_{n=1}^{N_{\rm{b}}} \left( \bar{q}_{k,i} \lambda_{k,n} +\sigma_{}^2\right) ,
\end{align}
and 
\begin{equation}\begin{aligned} \label{inv}
		\left( \bar{q}_{k,i}{\tilde{\bf{R}}}_{{\rm{JM}},k} + \sigma^2{\bf{I}}_{N_{\rm{b}}}\right) ^{-1} 
		& =\left[ {\bf{V}}_k\left( \bar{q}_{k,i}{\boldsymbol{\Lambda}}_k+ \sigma^2{\bf{I}}_{N_{\rm{b}}}\right) {\bf{V}}^{\rm{H}}_k\right] ^{-1},\\
		& = {\bf{V}}_k\left( \bar{q}_{k,i}{\boldsymbol{\Lambda}}_k+ \sigma^2{\bf{I}}_{N_{\rm{b}}}\right) ^{-1}{\bf{V}}^{\rm{H}}_k .
\end{aligned}\end{equation}
By utilizing Eqs. \eqref{det} and \eqref{inv} in Eqs. \eqref{pdf h1} and \eqref{joint pdf}, the log-likelihood function in Eq. \eqref{joint pdf} can be further expressed as Eq. \eqref{log likelihood}, shown at the bottom of this page.\setcounter{equation}{20}
According to Eq. \eqref{log likelihood}, 
\begin{align}
	\frac{\partial 	\ln f\left( {\bar{\bf{y}}}_{k ,i};{q}_k,\mathcal{H}_1\right) }{\partial {q}_k} =  {-\sum_{n=1}^{N_{\rm{b}}}\frac{\lambda_{k,n}}{{q}_{k}\lambda_{k,n} +\sigma^2}}
	+\sum_{n=1}^{N_{\rm{b}}}
	\frac{\lambda_{k,n}|{\bar{\bf{y}}}^{\rm{H}}_{k,i}  {{\bf{v}}}_{k,n} |^2}
	{\left( {q}_{k}\lambda_{k,n}+\sigma^2\right) ^2} .
\end{align}
Therefore, 
\begin{align}
	\frac{\partial \ln f\left( {\bar{\bf{y}}}_{k,i };{q}_k,\mathcal{H}_1\right) }{\partial q_k}\bigg|_{{q}_k = 0} = {-{\frac{1}{\sigma^2}}{\sum_{n=1}^{N_{\rm{b}}}}{\lambda_{k,n}}}
	+\frac{1}{\sigma^4}\sum_{n=1}^{N_{\rm{b}}}
	{\lambda_{k,n}|{\bar{\bf{y}}}^{\rm{H}}_{k,i}  {{\bf{v}}}_{k,n} |^2}.
\end{align}
Combining with Eq. \eqref{lmpt-1}, we get the LMPT-based jamming detector as
\begin{equation} \label{LMPT detector}
	T_{\rm{LMPT}}\left( \left\lbrace {\bar{\bf{y}}}_{k,i }\right\rbrace_{i=2}^{\tau} \right)  = \sum_{i=2}^{\tau}\sum_{n=1}^{N_{\rm{b}}}
	{\lambda_{k,n}|{\bar{\bf{y}}}^{\rm{H}}_{k,i}  {{\bf{v}}}_{k,n} |^2} \stackrel{\mathcal{H}_{1}}{\underset{\mathcal{H}_{0}}{\gtrless}}\gamma'_{\rm{LMPT}},
\end{equation} 
where $\gamma'_{\rm{LMPT}}$ is the decision threshold of the LMPT detector.
The computational complexity of the proposed  jamming detector is analyzed as follows. 
The complexity for the eigen-decomposition of ${\tilde{\bf{R}}}_{{\rm{JM}},k}$ is typically on
the order of $\mathcal{O}({N^3_{\rm{b}}})$.
The computation of the inner-product of two $N_{\rm{b}}$-length vectors has a complexity of the order $\mathcal{O}({N_{\rm{b}}})$.
Since $(\tau-1){N_{\rm{b}}}$ times vector multiplications are involved, the computational complexity of the LMPT detector is $\mathcal{O}\left({N^3_{\rm{b}}}+{N^2_{\rm{b}}}\tau\right)$.
The simplicity of the LMPT statistics allows performance analysis of the detector.



\vspace{0.5em}
\subsection{Performance Analysis of the LMPT-Based  Scheme} \label{Performance Analysis of LMPT}
In the following, we proceed to analyze the distribution of the LMPT statistic and evaluate the performance of the proposed LMPT-based jamming detection scheme assisted by channel statistics. This includes deriving the probability of detection $P_{\rm{D}}$ and the probability of false alarm $P_{\rm{FA}}$.
The following theorem provides the distribution of $T_{\rm{LMPT}}$.

\begin{theorem} \label{theorem1}The PDF of $T_{\rm{LMPT}}$ under ${\mathcal{H}}_0$ is
	\vspace{-0.2cm}
	\begin{align} \label{PDF h0}
		f_{T}(x;\mathcal{H}_0)&=\sum_{m=0}^{\infty} \frac{a_{k,m}}{\Gamma\left[(\tau-1)\rho_k+m\right]({2 \beta_k})^{(\tau-1)\rho_k+m}} \nonumber
		\\&\hspace{2.8cm} \times x^{(\tau-1)\rho_k+m-1} e^{-\frac{x}{2 \beta_k}},
	\end{align}
	where $\beta_{k}=\frac{\rho_k}{2\sum_{j=1}^{\rho_k}\left(\lambda_{k,n}\sigma^2 \right) ^{-1}}$, 
	\begin{equation*}
		\label{para_lemma0}
		\begin{aligned}
			&a_{k,0}=\prod_{j=1}^{\rho_k}\left( \frac{2\beta_k}{\lambda_{k,n}\sigma^2}\right) ^{\tau-1},
			b_{k,m}=2\sum_{j=1}^{\rho_k}\left(1-  \frac{2\beta_k}{\lambda_{k,n}\sigma^2}\right)^m,\\
			&a_{k,m}= \frac{1}{2m}\sum_{r=0}^{m-1}b_{k,m-r}a_{k,r},
			\forall m>1.
		\end{aligned}
	\end{equation*} 
The PDF of $T_{\rm{LMPT}}$ under  \(\mathcal{H}_1\) is
\begin{equation} \label{PDF h1}
	f_{T}(x;\mathcal{H}_1)=\sum_{m=0}^{\infty} \frac{\bar a_{k,m}}{\Gamma\left(\varphi_k+m\right)\left( {2 \bar\beta_k}\right) ^{\varphi_k+m}} x^{\varphi_k+m-1} e^{-\frac{x}{2 \bar\beta_k}},
\end{equation}
where $\varphi_k = {{rank}}\left( {\bf{B}}_k \right) $, 
\begin{align} \label{BK}
	{\bf{B}}_k = {{\bf{R}}}^{1/2}_{{\bar{\bf{y}}}_{k }}
	\left({\mathbf{I}}_{\tau-1} \otimes {\tilde{\bf{R}}}_{{\rm{JM}},k}\right) 
	{{\bf{R}}}^{1/2}_{{\bar{\bf{y}}}_{k }}, 
\end{align}
\begin{equation*}
	\label{para_lemma1}
	\begin{aligned}
		&\bar{\beta}_{k}=\frac{\varphi_k}{2\sum_{j=1}^{\varphi_k}\epsilon_{k,j}^{-1}},
		&&\bar{a}_{k,0}=\prod_{j=1}^{\varphi_k}\frac{2\bar \beta_k}{\epsilon_{k,j}},\\
		&\bar b_{k,m}=2\sum_{j=1}^{\varphi_k}\left(1-\frac{2\bar\beta_k}{\epsilon_{k,j}}\right)^m,
		&&\bar a_{k,m}=\frac1{2m}\sum_{r=0}^{m-1}\bar b_{k,m-r}\bar a_{k,r},\\
		&&&\hspace{2.35cm}\forall m>1.
	\end{aligned}
\end{equation*} 
and $\epsilon_{k,1},\ldots,\epsilon_{k,\varphi_k}$ are the positive eigenvalues of ${\bf{B}}_k$.
\end{theorem}
\begin{proof}
See Appendix A.
\end{proof}

Based on the PDFs in Theorem \ref{theorem1}, the following theorem gives analytic expressions for the probability of false alarm and the probability of detection.
\begin{theorem} \label{theorem2}
	With the LMPT-based detector for jamming detection,
	the false alarm probability and the detection probability in the $k$-th subband for a
	given threshold $\gamma'_{\rm{LMPT}}$ are derived as follows
	\begin{equation}\begin{aligned} \label{PFA}
			P_{\rm{FA}} 
			& = \sum_{m=0}^{\infty} a_{k,m} \frac{\Gamma\left((\tau-1)\rho_k+m,\frac{\gamma'_{\rm{LMPT}}}{2\beta_k}\right)}{\Gamma\left( \left(\tau-1\right)\rho_k+m\right) },
	\end{aligned}\end{equation}
\begin{equation}\begin{aligned} \label{PD}
		P_{\rm{D}} 
		& = \sum_{m=0}^{\infty}\bar a_{k,m} \frac{\Gamma\left(\varphi_k+m,\frac{\gamma'_{\rm{LMPT}}}{2\bar\beta_k}\right)}{\Gamma\left(\varphi_k+m\right)} ,
\end{aligned}\end{equation}
where $\Gamma\left( \cdot,\cdot\right) $ is the upper incomplete Gamma function.
\end{theorem}
\begin{proof}
	By using the definitions of the detection probability and false alarm probability, along with the PDFs  given in Eqs. \eqref{PDF h0} and \eqref{PDF h1}, the above expression can be easily derived.
\end{proof}

The analytical expression for the false alarm probability relies on the parameters $\left\lbrace  a_{k,m} \right\rbrace_{m=1}^{\infty} $ and $\beta_k $, which are computed based on \(\left\lbrace \lambda_{k,n}\sigma^2\right\rbrace _{n=1}^{\rho_k}\).
Thus, with the known receiver noise level and channel statistics of the jammer, 
we can determine the threshold ${\gamma'_{\rm{LMPT}}}$ for the desired $P_{\rm{FA}} $ level using the derived false alarm probability expression.
By noticing that  the  false alarm probability decreases as the decision threshold increases,
 the bisection method can be employed.
\section{MMSE-Based Two-Step Channel Estimation For the Jammer and Legitimate Users }\label{MMSE-based Spatial-Domain CE}
\vspace{-0.2em}

In this section, we study the estimation of both jammer and user channels. Due to the unknown nature of the random jammer pilot, there are unknown parameters, which are the jamming pilot's inner-products, in the channel estimation problems. 
The estimation of the jamming channel is actually a joint estimation of ${{{\bf{h}}}^{\rm{JM}}_{k}}$ and $\boldsymbol{\alpha}_k$. Similarly, the estimation of the user channel is also a joint estimation of ${{\bf{h}}}_{k}$ and $|{\alpha}_{k,1}|$.

To begin with, we show that there is 
an ambiguity problem in the estimation of ${{{\bf{h}}}^{\rm{JM}}_{k}}$ and ${\boldsymbol{\alpha}}_k$ with respect to one phase.
For an arbitrary angle $\theta$,  $\left( {{{\bf{h}}}^{\rm{JM}}_{k}}e^{-j\theta},{\boldsymbol{\alpha}}_k e^{-j\theta}\right) $ will produce the same observation signals, referred to Eq. \eqref{projected signal},  as  $\left( {{{\bf{h}}}^{\rm{JM}}_{k}},{\boldsymbol{\alpha}}_k \right) $, and
the properties of $\left(  {{{\bf{h}}}^{\rm{JM}}_{k}},{\boldsymbol{\alpha}}_k  \right) $ are invariant under the phase shift. Consequently, 
there is a phase that is not observable from the model, and the estimation problem has ambiguity (in terms of a phase shift).
Note that there is no ambiguity in the estimation of the product $\alpha_{k,i}{{{\bf{h}}}^{\rm{JM},*}_{k}}$.
By combining the angle of $\alpha_{k,2}$ with ${{{\bf{h}}}^{\rm{JM}}_{k}}$, i.e., by redefining $\alpha_{k,2}$ as $|\alpha_{k,2}|$ and ${{{\bf{h}}}^{\rm{JM}}_{k}}$ as ${{{\bf{h}}}^{\rm{JM}}_{k}}e^{ -j \angle \alpha_{k,2} }$,  the ambiguity of the estimation problem vanishes with the condition that the phase of the new $\alpha_{k,2}$ is zero.
Further define 
	\begin{align} \label{redefined alpha}
		\bar{\boldsymbol{\alpha}}_{k} = \left[\left| {\alpha}_{k,2}\right| ,\left| {\alpha}_{k,3}\right|e^{j \theta_{k,3}},\ldots, \left| {\alpha}_{k,\tau}\right|e^{j \theta_{k,\tau}}\right]^{\rm{T}} = {\boldsymbol{\alpha}}_{k}e^{ -j \angle \alpha_{k,2} },
	\end{align}
where $\theta_{k,i} = \angle \alpha_{k,i} - \angle \alpha_{k,2},i \in \vmathbb{t}\setminus \left \{ 2 \right \},$ denotes the angle difference of the jamming pilots' inner-product.

We then propose a two-step estimation scheme using the projected observation vectors. The first step is to estimate the norms and phase differences of the jamming pilot's inner-products, as shown in Section \ref{estimation of pilot}. The second step, detailed in  Sections \ref{jamming channel estimation}  and \ref{user channel estimation}, is to estimate  the channels based on the estimated values of the inner-products.



\vspace{0.3em}
\subsection{Multiple Projected Observation Vectors Based MMSE Jamming Channel Estimation}\label{jamming channel estimation}

In this subsection, we work on the estimation of the jamming channel.
With the stacked vector consisting of all projected observation vectors  
$\bar{\mathbf{y}}_k$ in Eq. \eqref{projected signal}, the linear MMSE estimation can be expressed as ${\hat{\bf{h}}}^{{\rm{JM}},*}_{k} = {\mathbf{R}}_{{{{\bf{h}}}^{{\rm{JM}},*}_{k}},\bar{\bf{y}}_{k}}{\mathbf{R}}^{-1}_{\bar{\bf{y}}_{k}}\bar{\mathbf{y}}_k$,
where ${\mathbf{R}}_{{{\bf{h}}}^{{\rm{JM}},*}_{k},\bar{\bf{y}}_{k}} = \mathbb{E}\left[{{{\bf{h}}}^{{\rm{JM}},*}_{k}}{\bar{\bf{y}}}_{k}^{\rm{H}}\right]$ denotes the mutual covariance matrix between the conjugate of the jamming channel and the stacked vector, and ${\mathbf{R}}^{}_{\bar{\bf{y}}_{k}}$ is the self-covariance matrix of the vector ${\bar{\bf{y}}}_{k}$, as shown in Eq. \eqref{all-cov}.

For the mutual covariance matrix ${\mathbf{R}}_{{{{\bf{h}}}^{{\rm{JM}},*}_{k}},\bar{\bf{y}}_{k}}$, we have
\begin{align} \label{cross-cor}
		\quad {\mathbf{R}}_{{{{\bf{h}}}^{{\rm{JM}},*}_{k}},\bar{\bf{y}}_{k}}
		&= \mathbb{E}\left[{{\bf{h}}}^{{\rm{JM}},*}_{k} \left (  {\sqrt{\tau q_k}}\bar{\boldsymbol{\alpha}}_{k}\otimes {{\bf{U}}}^{\rm{H}}_{\rm{JM}} 
		{{{\bf{h}}}^{{\rm{JM}},*}_{k}}\right )^{\rm{H}} \right]
		\nonumber \\&\hspace{2cm}+  \mathbb{E}\left[{{{\bf{h}}}^{{\rm{JM}},*}_{k}}\left ({\bf{n}}^{\rm{H}}_{k,2},\ldots,{\bf{n}}^{\rm{H}}_{k,\tau}  \right ) \right],\nonumber \\
		& = {\sqrt{\tau q_k}}\bar{\boldsymbol{\alpha}}_{k}^{\rm{H}}\otimes \left({{\bf{R}}}^{*}_{{\rm{JM}},k} {{\bf{U}}}_{\rm{JM}}  \right) .
\end{align}
From Eqs. \eqref{all-cov} and \eqref{cross-cor}, 
the MMSE estimate of the jammer's channel is 
\begin{equation}\begin{aligned} \label{jam MMSE result}
		{\hat{{\bf{h}}}_{k}^{{\rm{JM}}}} 
		& = \left(
		{\sqrt{\tau q_k}}{\bar{\boldsymbol{\alpha}}}_k^{\rm{H}}\otimes \left({{\bf{R}}}^{*}_{{\rm{JM}},k} {{\bf{U}}}_{\rm{JM}} \right) \right.  \\
		&\quad \times \left. \left [  \tau q_k\left( {\bar{\boldsymbol{\alpha}}}_k{\bar{\boldsymbol{\alpha}}}_k^{\rm{H}}\right) \otimes\tilde{\mathbf{R}}_{{\rm{JM}},k}+\sigma^2{\bf{I}}_{\left( \tau-1\right)  N_{\rm{b}}}
		\right ] ^{-1}
		{\bar{\bf{y}}}_{k}  \right)^{*}.\\
\end{aligned}\end{equation}
The jamming channel estimate $	{\hat{{\bf{h}}}_{k}^{{\rm{JM}}}} $ relies on the redefined  jamming pilot's inner-product vector.
The estimation of  $\bar{\boldsymbol{\alpha}}_{k} $  is deferred to  Section \ref{estimation of pilot}.

\vspace{1em}
\subsection{MMSE Criterion-Based Users' Channel Estimation}\label{user channel estimation}

From Eq. \eqref{projected signal used},  the projected observation vector corresponding to the used pilot ${\bar{\bf{y}}}_{k,1 }$ contains channel information about the legitimate user, so we
  use this vector to estimate the user channel. The desired  mutual covariance matrix between the conjugate of the user's channel and  the projected observation vector ${\bar{\bf{y}}}_{k,1 }$ is first calculated
\begin{equation}\begin{aligned}
		{\mathbf{R}}_{{{\bf{h}}}^{*}_{k},{\bar{\bf{y}}}_{k,1 }}&= \mathbb{E}\left[{{\bf{h}}}^{*}_{k}{\bar{\bf{y}}}_{k,1 }^{\rm{H}}\right]= {\sqrt{\tau p_t}}{\mathbf{R}}_{k}^*{{\bf{U}}}_{k},
\end{aligned}\end{equation}
and the self-covariance matrix of the vector ${\bar{\bf{y}}}_{k,1 }$ is
\begin{equation}\begin{aligned}
		{\mathbf{R}}_{{\bar{\bf{y}}}_{k,1 }}& = \mathbb{E}\left[{\bar{\bf{y}}}_{k,1 }{\bar{\bf{y}}}_{k,1 }^{\rm{H}}\right],\\
		& = {{\tau p_t}}{{\bf{U}}}^{\rm{H}}_{k}{\mathbf{R}}_{k}^*{{\bf{U}}}_{k}+ {{\tau q_k}}\left| {\alpha}_{k,1}\right|^2 \tilde{\mathbf{R}}_{{\rm{JM}},k}
		+\sigma^2{\bf{I}}_{N_{\rm{b}}}.
\end{aligned}\end{equation}
Hence, the linear MMSE channel estimate for user $k$ is
\begin{align}
		{\hat{{\bf{h}}}_{k}} 
		&= 
		\left({\mathbf{R}}_{{{\bf{h}}}^{*}_{k},{\bar{\bf{y}}}_{k,1 }}{\mathbf{R}}^{-1}_{{\bar{\bf{y}}}_{k,1 }}{\bar{\bf{y}}}_{k,1 }\right)^{*},\nonumber\\
		&={\sqrt{\tau p_t}}{\mathbf{R}}_{k}{{\bf{U}}}^{*}_{k}\left[{{\tau p_t}}{{\bf{U}}}^{\rm{T}}_{k}{\mathbf{R}}_{k}{{\bf{U}}}^{*}_{k}\right.  \nonumber
		\\ & \hspace{1cm} \left.+ {{\tau q_k}}\left| {\alpha}_{k,1}\right|^2\tilde{\mathbf{R}}^{*}_{{\rm{JM}},k}
		+\sigma^2{\bf{I}}_{N_{\rm{b}}}\right]^{-1}{\bar{\bf{y}}}^{*}_{k,1 }.
\end{align}
Similarly, the estimation depends on $|\alpha_{k,1}|^2$ and  we refer to Section \ref{estimation of pilot}
for the estimation of $|\alpha_{k,1}|^2$.

\vspace{1em}
\subsection{Estimation Scheme for the Jamming Pilot's Inner-Product }\label{estimation of pilot}
\vspace{-1em}
The jamming pilot's inner-product value  involved in channel estimation comprises two components: $|{\alpha}_{k,1}|$ and  the redefined $\bar{\boldsymbol{\alpha}}_{k} $ in Eq. \eqref{redefined alpha}.
$|{\alpha}_{k,1}|$ affects the estimation of the legitimate user's channel, while 
$\bar{\boldsymbol{\alpha}}_{k}$ affect the estimation of the jammer's channel.
From Eq. \eqref{projected signal}, for $i \in \vmathbb{t}, $
\begin{align}\label{inner product of y2}
	&\|\bar{\mathbf{y}}_{k,i}\|^2=\tau q_k|\alpha_{k,i}|^2{\tilde{\bf{h}}}^{{\mathrm{JM}},\mathrm{H}}_{k}{\tilde{\bf{h}}}^{\mathrm{JM}}_{k}+{v}_{k,i},
\end{align}
where 
$
	{v}_{k,i} = 
	2\sqrt{\tau q_k} \Re\left[  \alpha_{k,i}\mathbf{n}^{\rm{H}}_{k,i}{\tilde{\bf{h}}}^{{\mathrm{JM}}}_{k} \right] + \mathbf{n}^{\rm{H}}_{k,i}\mathbf{n}_{k,i}.
$
Eq. \eqref{inner product of y2}  can be further written as 
\begin{align} \label{slight variations}
	&\|\bar{\mathbf{y}}_{k,i}\|^2=\tau q_k|\alpha_{k,i}|^2{\rm{tr}}\left( {\tilde{\bf{R}}}_{{\mathrm{JM}},k}\right) +\tilde{{v}}_{k,i},
\end{align}
where $\tilde{{v}}_{k,i} = {{v}}_{k,i} + \Delta_{k,i} $ and 
\begin{align}\label{the differece of 1st and average}
	\Delta_{k,i} =  \tau q_k|\alpha_{k,i}|^2{\tilde{\bf{h}}}^{{\mathrm{JM}},\mathrm{H}}_{k}{\tilde{\bf{h}}}^{\mathrm{JM}}_{k}
	-\tau q_k|\alpha_{k,i}|^2{\rm{tr}}\left( {\tilde{\bf{R}}}_{{\mathrm{JM}},k}\right) .
\end{align}
We can do similar analysis of Eq. \eqref{projected signal used} to get
$\|\bar{\mathbf{y}}_{k,1}\|^2=\tau q_k|\alpha_{k,1}|^2{\rm{tr}}\left( {\tilde{\bf{R}}}_{{\mathrm{JM}},k}\right) +\tilde{{v}}_{k,1}$,
where
\begin{align} \label{random and noise}
	\tilde{{v}}_{k,1} =& \tau {p_{t}}{\tilde{\bf{h}}}^{\rm{H}}_k{\tilde{\bf{h}}}_k + \Delta_{k,1} 
	+ 2 \tau\sqrt{p_t q_k}\Re \left[ \alpha_{k,1}{\tilde{\bf{h}}}^{\rm{H}}_k {\tilde{\bf{h}}}^{\mathrm{JM}}_{k}   \right]  + \left\| \mathbf{n}_{k,1}\right\| ^2
	\nonumber
	\\&+ 
	2 \sqrt{\tau p_t} \Re \left[ \mathbf{n}^{\rm{H}}_{k,1} {\tilde{\bf{h}}}_k   \right] 
	+	2\sqrt{\tau q_k}  \Re \left[ \alpha_{k,1}\mathbf{n}^{\rm{H}}_{k,1}  {\tilde{\bf{h}}}^{\mathrm{JM}}_{k}   \right]
	, \nonumber
\end{align}
${\tilde{\bf{h}}}_k \triangleq {\bf{U}}_k^{\rm{H}}{\bf{h}}_k^* \sim\mathcal{CN}\left( {\bf{0}} , {\tilde{\bf{R}}}_k \right) $ with ${\tilde{\bf{R}}}_k={\bf{U}}_k^{\rm{H}}{\bf{R}}^*_k{\bf{U}}_k$, and $\Delta_{k,1}$ is of the same form as defined in Eq. \eqref{the differece of 1st and average}.
Now, considering $\sum_i|\alpha_{k,i}|^2=1$, we have
\begin{align}
	\|\bar{\mathbf{y}}_{k,1}\|^2=-\tau q_k{\rm{tr}}({\tilde{\bf{R}}}_{{\mathrm{JM}},k})\sum_{i\in \vmathbb{t}}|\alpha_{k,i}|^2+\tilde{{v}}_{k,1}.
\end{align}
Thus, we obtain the linear model: ${\bf{b}}_k = {{\bf{C}}}_k {\bf{x}}_{k} +\tilde{\bf{v}}_{k}$, 
where ${\bf{b}}_k = [\|\bar{\mathbf{y}}_{k,1}\|^2,\ldots,\|\bar{\mathbf{y}}_{k,\tau}\|^2]^{\rm{T}}$, ${\bf{x}}_{k} = \left[|\alpha_{k,2}|^2,\ldots,|\alpha_{k,\tau}|^2 \right] ^{\rm{T}}$, $\tilde{\bf{v}}_{k} = \left[ \tilde{{v}}_{k,1},\ldots,\tilde{{v}}_{k,\tau}\right]^{\rm{T}} $,
and ${{\bf{C}}}_k = \tau q_k{\rm{tr}}({\tilde{\bf{R}}}_{{\mathrm{JM}},k}) \left[-{\bf{1}}_{\tau-1}, {\bf{I}}_{\tau-1} \right]^{\rm{T}}$.
The linear MMSE estimation of ${\bf{x}}_{k}$ is
\begin{align} \label{norm estimation}
	{\hat{\bf{x}}}_{k} = \mathbb{E}\left[ {\bf{x}}_{k}{\bf{b}}^{\rm{T}}_k \right] \left( \mathbb{E}\left[ {\bf{b}}_k {\bf{b}}^{\rm{T}}_k \right] \right) ^{-1} {\bf{b}}_k.
\end{align}
Considering that $|{\alpha}_{k,i}|^2$ is non-negative, we propose the following estimate based on Eq. \eqref{norm estimation}
\begin{align} \label{40}
	|\hat{\alpha}_{k,i}|^2 = \max(0, {\hat{{x}}}_{k,i}),
\end{align}
where ${\hat{{x}}}_{k,i}$ is the $i$-th element of the vector ${\hat{\bf{x}}}_{k}$.
Since the random/noise terms are uncorrelated with the jamming pilot's inner-product, $\mathbb{E}\left[ {\bf{x}}_{k}{\bf{b}}^{\rm{T}}_k \right]$ and $\mathbb{E}\left[ {\bf{b}}_k {\bf{b}}^{\rm{T}}_k \right]$ can be expressed as
\begin{align}
	\mathbb{E}\left[ {\bf{x}}_{k}{\bf{b}}^{\rm{T}}_k \right] &= \mathbb{E}\left[ {\bf{x}}_{k}{\bf{x}}^{\rm{T}}_k \right]{{\bf{C}}}_k^{\rm{T}} + \mathbb{E}\left[ {\bf{x}}_{k} \right]
	\mathbb{E}\left[ \tilde{\bf{v}}^{\rm{T}}_{k} \right], \label{41}\\ 
	\mathbb{E}\left[ {\bf{b}}_k {\bf{b}}^{\rm{T}}_k \right] &= {{\bf{C}}}_k\mathbb{E}\left[ {\bf{x}}_{k}{\bf{x}}^{\rm{T}}_k \right]{{\bf{C}}}_k^{\rm{T}} + {{\bf{C}}}_k\mathbb{E}\left[ {\bf{x}}_{k} \right]
	\mathbb{E}\left[ \tilde{\bf{v}}^{\rm{T}}_{k} \right] \nonumber \\
	&\hspace{0.5cm}+ \mathbb{E}\left[ \tilde{\bf{v}}_{k} \right]\mathbb{E}\left[ {\bf{x}}_{k}^{\rm{T}} \right]{{\bf{C}}}_k^{\rm{T}}+
	\mathbb{E}\left[\tilde{\bf{v}}_{k} \tilde{\bf{v}}^{\rm{T}}_{k} \right] \label{42}.
\end{align}
The problem now shifts to determining the expectations and self-covariance matrices of ${\bf{x}}_{k}$ and $\tilde{\bf{v}}_{k}$.
The following lemmas give the expressions for these values.

\begin{lemma} \label{lemma1}
	The expectation of ${\bf{x}}_{k}$ is 
	\begin{align}\label{expectation of xk}
		\mathbb{E}\left[ {\bf{x}}_{k}\right]  = \frac{1}{\tau}{\bf{1}}_{\tau-1},
	\end{align}
	and the correlation matrix is 
	\begin{align} \label{covariace of xk}
		\mathbb{E}\left[ {\bf{x}}_{k}{\bf{x}}^{\rm{T}}_k \right] = \frac{1}{\tau^2} {\bf{1}}{\bf{1}}^{\rm{T}} + \frac{1}{\tau^2} {\bf{I}}_{\tau-1}.
	\end{align}
\end{lemma}
\begin{proof}
	See Appendix B.
\end{proof}

\begin{lemma} \label{lemma2}
	The expectation of $\tilde{\bf{v}}_{k}$ is 
	\begin{align} \label{expectation of vk}
		\mathbb{E}\left[ \tilde{\bf{v}}_{k} \right]
		=  \tau {p_{t}}{\rm{tr}}\left({\tilde{\bf{R}}}_k \right)
		{\bf{e}}_1
		+N_{\rm{b}}\sigma^2 
		{\bf{1}}_{\tau},
	\end{align}
	and the  correlation matrix can be expressed as
	\begin{align}
		\mathbb{E}\left[\tilde{\bf{v}}_{k} \tilde{\bf{v}}^{\rm{T}}_{k} \right]
		&= \begin{pmatrix}
			\varrho_{1,1}  &\varrho_{1,2}{\bf{1}}_{\tau-1}  \\
			\varrho_{2,1}{\bf{1}}_{\tau-1}^{\rm{T}}  &\varrho_{2,2}{\bf{1}}{\bf{1}}^{\rm{T}}
		\end{pmatrix},
	\end{align}
	where $\varrho_{1,1} = \tilde{\varrho}+  2{\tau p_t q_k}\left| {\rm{tr}}\left( {\tilde{\bf{R}}}_{{\rm{JM}},k} {\tilde{\bf{R}}}_{k} \right)\right| + 2{q_k\sigma^2}{\rm{tr}}\left( {\tilde{\bf{R}}}_{{\rm{JM}},k}  \right) $,
	$$
	\tilde{\varrho} = \tau^2p_t^2\Theta_{k}  + 2q_k^2\left\|{\tilde{\bf{R}}}_{{\rm{JM}},k} \right\|_{\rm{F}}^2
	+{\small{\left(  N_{\rm{b}}+1\right)}}\sigma^2\left( 2\tau p_t{\rm{tr}}\left( {\tilde{\bf{R}}}_{k}\right)  + N_{\rm{b}}\sigma^2 \right) , $$
	$$
	\varrho_{1,2} = \varrho_{2,1} = N_{\rm{b}}\sigma^2 \left( \tau p_t{\rm{tr}}\left( {\tilde{\bf{R}}}_{k}\right)  + N_{\rm{b}}\sigma^2 \right)    +q_k^2\left\|{\tilde{\bf{R}}}_{{\rm{JM}},k} \right\|_{\rm{F}}^2,
	$$
	$$
	\varrho_{2,2} =  N^2_{\rm{b}}\sigma^4  +q_k^2\left\|{\tilde{\bf{R}}}_{{\rm{JM}},k} \right\|_{\rm{F}}^2,
	$$
	with $\Theta_{k}  =  \left| {\rm{tr}}\left( {\tilde{\bf{R}}}_k \right)  \right| ^2 +\left\|{\tilde{\bf{R}}}_k \right\|_{\rm{F}}^2$. 
\end{lemma}
\begin{proof}
	See Appendix C.
\end{proof}

The proposed estimate of $|{\alpha}_{k,i}|^2$ can be obtained by using results in
Lemmas \ref{lemma1} and \ref{lemma2}, 
in  Eqs. \eqref{41}  \eqref{42} and then Eqs. \eqref{norm estimation} and \eqref{40}.

Furthermore, for the phase differences $\left\lbrace \theta_{k,i},i\in\vmathbb{t}\setminus \left\{ 2\right\}  \right\rbrace $ in Section \ref{jamming channel estimation}, we consider an estimation method based on the following bilinear form
\begin{equation}\begin{aligned} \label{angle diff estimation}
		\hat{\theta}_{k,i}=\angle(\bar{\mathbf{y}}_{k,2}^{\rm{H}}\mathbf{A}_{k,i}\bar{\mathbf{y}}_{k,i}), i\in\vmathbb{t} \setminus \left\{ 2\right\}  ,
\end{aligned}\end{equation}
where $\mathbf{A}_{k,i}\in \mathbb{C}^{N_{\rm{b}}\times N_{\rm{b}}}$ is a positive semi-definite matrix. 
From Eq. \eqref{projected signal}, Eq. \eqref{explicit form of angle estimation} can be obtained as shown at the bottom of the next page.
We consider  a heuristic design of $\mathbf{A}_{k,i}$ with low complexity: $\mathbf{A}_{k,i} = \frac{1}{{ N_{\rm{b}}}} {\bf{I}}_{ N_{\rm{b}}}$, $\forall  i\in\vmathbb{t}\setminus \left\{ 2\right\} $.
It can be shown by using the orthogonality among the beam channel and noise and  among noise vectors on different pilots that $\frac{\bar{\mathbf{y}}_{k,2}^{\rm{H}}\bar{\mathbf{y}}_{k,i} }{{ N_{\rm{b}}}}$  converges almost surely to $\tau q_k |\alpha_{k,2}\alpha_{k,i} | e^{j\theta_{k,i}}
       \frac{\left\|{\tilde{\bf{h}}}^{{\mathrm{JM}}}_{k} \right\|^2 }{{ N_{\rm{b}}}}$ when ${ N_{\rm{b}}}$ approaches infinity.

\begin{figure*}[!b]
	\begin{minipage}[b]{0.5\textwidth}
		\hspace{-0.4cm}
		\includegraphics[scale=0.7]{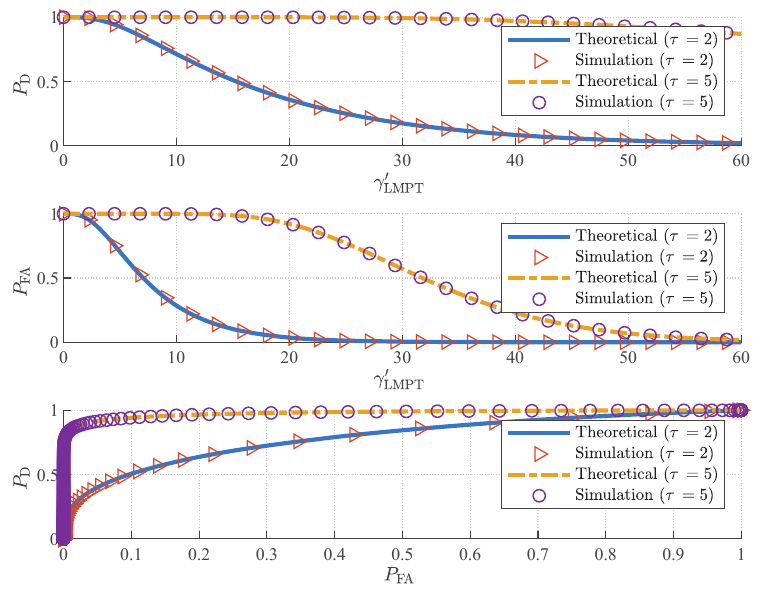}
		\caption{Theoretical and simulation performance of the proposed 
			LMPT-based jamming detection scheme for $\tau = 2,5$.}
		\label{fig2}
	\end{minipage}
	\begin{minipage}[b]{0.5\textwidth}
		\hspace{-0.4cm}
		\includegraphics[scale=0.66]{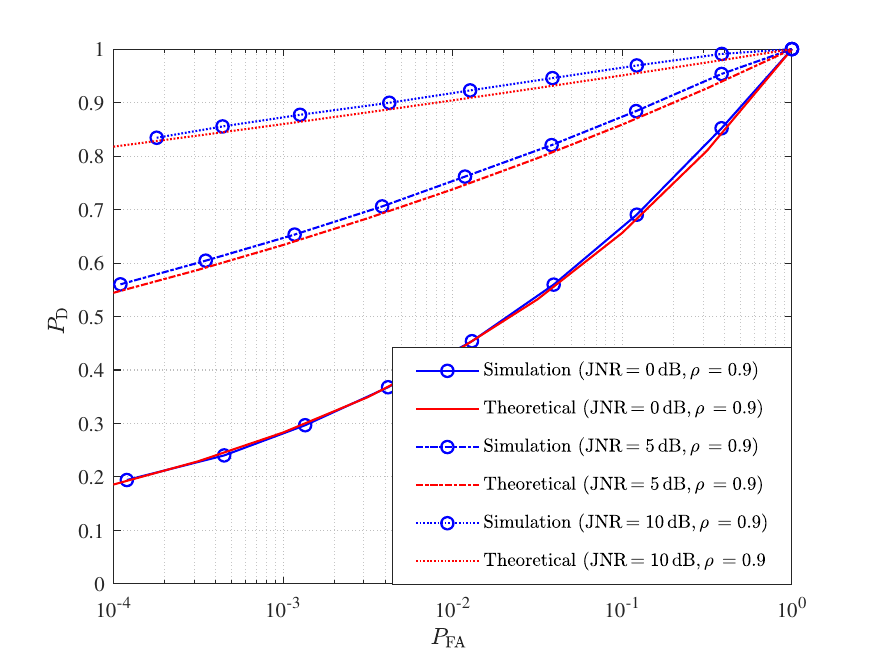}
		\caption{Theoretical and simulation ROC of the  LMPT-based jamming detection scheme for ${\rm{JNR}}=0$, $5$, $10\,{\rm{dB}}$.}
		\label{fig1}
	\end{minipage}
\end{figure*}

\begin{figure*}[hb]
	
	\newcounter{TempEqCnt4} 
	\setcounter{TempEqCnt4}{\value{equation}} 
	\setcounter{equation}{47} 
	\centering 
	\hrulefill 
	\begin{equation}\begin{aligned} \label{explicit form of angle estimation}
			\bar{\mathbf{y}}_{k,2}^{\rm{H}}\mathbf{A}_{k,i}\bar{\mathbf{y}}_{k,i} 
			&= \underbrace{ \tau q_k \alpha^{*}_{k,2}\alpha_{k,i} 
				{\tilde{\bf{h}}}^{{\mathrm{JM}},\mathrm{H}}_{k}\mathbf{A}_{k,i}{\tilde{\bf{h}}}^{{\mathrm{JM}}}_{k}
			}_{\rm{Desired\,\, Signal}}
			+ 
			\underbrace{ \tau q_k \alpha^{*}_{k,2}  {\tilde{\bf{h}}}^{{\mathrm{JM}},\mathrm{H}}_{k} \mathbf{A}_{k,i}{\mathbf{n}}_{k,i}  +
				\tau q_k \alpha^{*}_{k,i}   {\mathbf{n}}^{\rm{H}}_{k,2}\mathbf{A}_{k,i} {\tilde{\bf{h}}}^{{\mathrm{JM}}}_{k}}_{{\rm{Interference \,\, Signal}}} 
			+
			\underbrace{{\mathbf{n}}^{\rm{H}}_{k,2}\mathbf{A}_{k,i}{\mathbf{n}}_{k,i}}_{{\rm{Noise}}},
	\end{aligned}\end{equation}
	\vspace*{-2pt} 
\end{figure*}
\setcounter{equation}{\value{TempEqCnt4}}

So far, we can accomplish the estimation of $\bar{\boldsymbol{\alpha}}_{k} $ where the norm estimation is in Eq. \eqref{40} and the phase estimation is in Eq. \eqref{angle diff estimation}. What remains is the estimation of  the norm of the pilot's inner product $|\alpha_{k,1}|$. From Eq. \eqref{projected signal}, we can see that the projected observation vectors
$\left\lbrace {{{{\bar{\bf y}}}_{k,{i}}}},i\in\vmathbb{t}\right\rbrace $ contain information of $|\alpha_{k,1}|$ indirectly. From $|\alpha_{k,1}|^2 = 1- \sum_{i=2}^{\tau}|\alpha_{k,i}|^2$, $|\alpha_{k,1}|$ can be estimated to be
\setcounter{equation}{48}
\begin{equation}\begin{aligned} \label{alpha1 est1}
		\left| {{\grave{\alpha} _{k,1}}}\right | = \sqrt{ \left[ 1- \sum_{i\in \vmathbb{t}}\left| {{\hat\alpha _{k,i}}} \right |^2\right] ^{+}}.
\end{aligned}\end{equation}
Besides, 
the projected observation vector ${{{{\bar{\bf y}}}_{k,{1}}}}$ contain information of $|\alpha_{k,1}|$ directly. Specifically,  from Eq. \eqref{projected signal used}, we have
\begin{align}
		{\left\| {{{{\bar{\bf y}}}_{k,{1}}}} \right\|^2} 
		&\overset{(b)}{\approx} \tau {p_{t}}{\tilde{\bf{h}}}^{\rm{H}}_k{\tilde{\bf{h}}}_k
		+\tau {q_{k}}{\left| {{\alpha _{k,{1}}}} \right|^2}{\tilde{\bf{h}}}^{{\mathrm{JM}},\mathrm{H}}_{k}{\tilde{\bf{h}}}^{\mathrm{JM}}_{k} + {\bf{n}}_{k,{1}}^{\rm{H}}{{\bf{n}}_{k,{1}}},\nonumber\\
		&\overset{(c)}{\approx}\tau {p_{t}}{\rm{tr}}\left({\tilde{\bf{R}}} _k \right) + \tau {q_{k}}{\left| {{\alpha _{k,{1}}}} \right|^2}{\rm{tr}}\left( {\tilde{\bf{R}}}_{{\rm{JM}},k} \right) 
		+ {N_{{\rm{b}}}}\sigma^2,
\end{align}
where $(b)$ exploits the orthogonality of the beam channel to the noise, and $(c)$ utilizes the asymptotic property of massive MIMO \cite{akhlaghpasand2019jamming}. According to the above equation, another estimate of $|\alpha_{k,1}|$ is
\vspace{-1em}
\begin{equation}\begin{aligned} \label{alpha1 est2}
		\left| {{\acute\alpha _{k,1}}} \right | = 
		\sqrt {\frac{{{{\left[ {{{\left\| {{{{\bar{\bf y}}}_{k,1}}} \right\|}^2} - \tau {p_{t}}{\rm{tr}}\left({\tilde{\bf{R}}} _k \right) - {N_{{\rm{b}}}}\sigma^2} \right]}^ + }}}{{\tau {q_{k}}{\rm{tr}}\left( {\tilde{\bf{R}}}_{{\rm{JM}},k} \right)}}} .
\end{aligned}\end{equation}
To enhance the estimation accuracy of $|\alpha_{k,1}|$ by utilizing all received signals, we use a weighted combination of the two estimates Eqs. \eqref{alpha1 est1} and \eqref{alpha1 est2} to obtain
\begin{equation}\begin{aligned} \label{alpha1 est total}
		\left| {{\hat\alpha _{k,1}}} \right | = \epsilon  {{\grave{\alpha}  _{k,1}}} + \left ( 1-\epsilon   \right ) {{\acute{\alpha}  _{k,1}}}.
\end{aligned}\end{equation}
where $\epsilon\in\left[0,1 \right] $ denotes the weighting coefficient.

\section{Simulation and Discussion}
\label{Simulation}


\vspace{0.5em}
\subsection{Simulation Setup}
To evaluate the performance of the proposed channel-statistics-assisted jamming detection scheme and two-step channel estimation scheme, we consider a scenario where a BS equipped with $M_{\rm{BS}} = 64$ antennas serves $K = 5$ users each configured with $M_{\rm{UE}} = 16$ antennas, and a jammer configured with $M_{\rm{JM}} = 16$ antennas.
We normalize the noise variance to $ \sigma^2 = 1$, which implies that the signal-to-noise ratio (SNR) of the system is $p_t$.
Assume that the jammer allocates the same power in each subband, i.e., $q_k = q_t,k\in\mathbb{K}$, then the jamming-to-noise ratio (JNR) is $ q_t$. 
In the uplink beam training phase, analog beams of both the BS and users are selected sequentially from the standard discrete Fourier transform (DFT) codebook, denoted by  $\mathbf{F}=[\mathbf{f}_1,\ldots,\mathbf{f}_M]\in\mathbb{C}^{M\times M}$, where $M$ denotes the number of antennas on the device. The $q$-th element in the beam vector $\mathbf{f}_{n}$ is given by ${f}_{n,q}=\frac{1}{\sqrt{M}}\exp\left(\frac{\text{j}2\pi(n-1)(q-1)}{M}\right)$.
To model the channel correlation matrix, we adopt the exponential correlation model \cite{loyka2001channel,zhang2024interleaved}, commonly used in MIMO systems.
Specifically, the correlation between the $m$-th element and the $n$-th element of ${{\bf{h}}}^{{\rm{JM}}}_{k}$ in Eq. \eqref{general projection2} can be expressed as $\left[ {\bf{R}}_{{\rm{JM}},k}\right] _{m,n} = \mathbb{E}\left[{{{h}}}^{{\rm{JM}}}_{k,m}{{{h}}}^{{\rm{JM}},*}_{k,n} \right]  = \rho^{|m-n|}$, where  $\rho$, satisfing $0\le |\rho|\le 1$, is the correlation coefficient of adjacent antennas, and  ${{{h}}}^{{\rm{JM}}}_{k,m}$ denotes the equivalent channel coefficient from the jammer to the $m$-th antenna of the BS.
The matrices $ \mathbf{R}_{{\mathrm{UE},k}} $ and $\mathbf{R}_{{\mathrm{BS},k}}$, which are the covariance matrices of the user side and the BS side, are also modeled in the same way as $ {\bf{R}}_{{\rm{JM}},k}$.
Since ${{\bf{h}}}_{k} = {\rm{vec}}\left( {{\bf{H}}}_{k}\right) $, we have ${\bf{R}}_{k} = \mathbf{R}_{{\mathrm{UE},k}} \otimes \mathbf{R}_{{\mathrm{BS},k}}$ \cite{forenza2007simplified}.


\begin{figure}[!t]
	\vspace{-0.4cm}
	\hspace{-1cm}
	\includegraphics[scale=0.7]{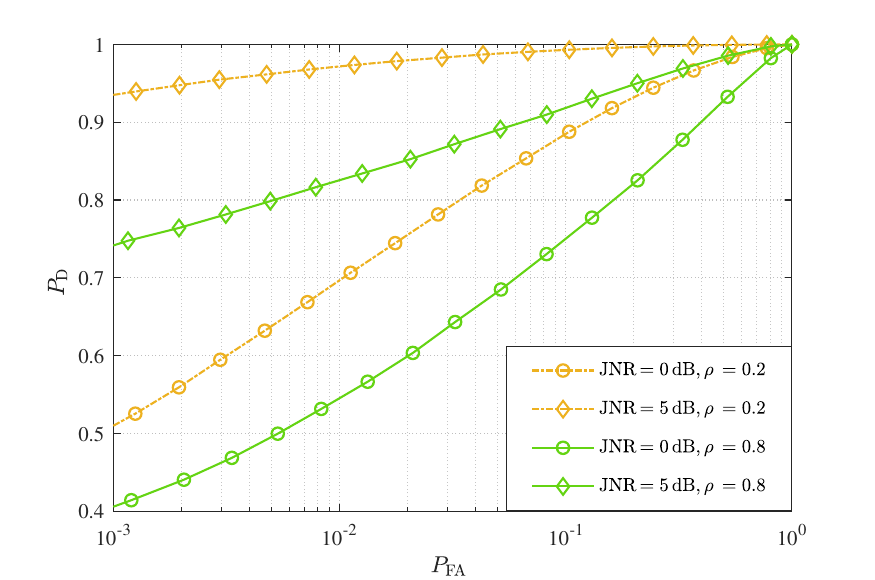}
	\vspace{-0.4cm}
	\caption{ROC of the proposed LMPT-based jamming detection scheme
		for $\rho=0.2$, $0.8$ and ${\rm{JNR}}=0$, $5\,{\rm{dB}}$.}
	\label{fig3}
	\vspace{-0.6cm}
\end{figure}

\begin{figure*}[!b]
	\vspace{-0.5cm}
	\hspace{-0.7cm}	
	\subfloat[Channel correlation $\rho = 0$]{
		\includegraphics[scale=0.45]{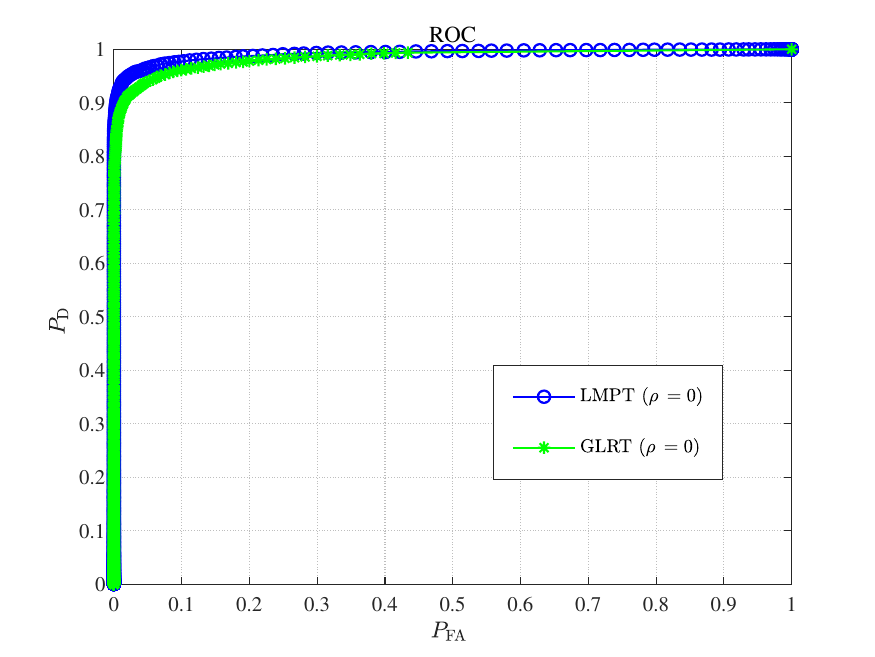}
		\label{fig4_1}}
	\hspace{-0.7cm}	
	\subfloat[Channel correlation $\rho = 0.5$]{
		\includegraphics[scale=0.45]{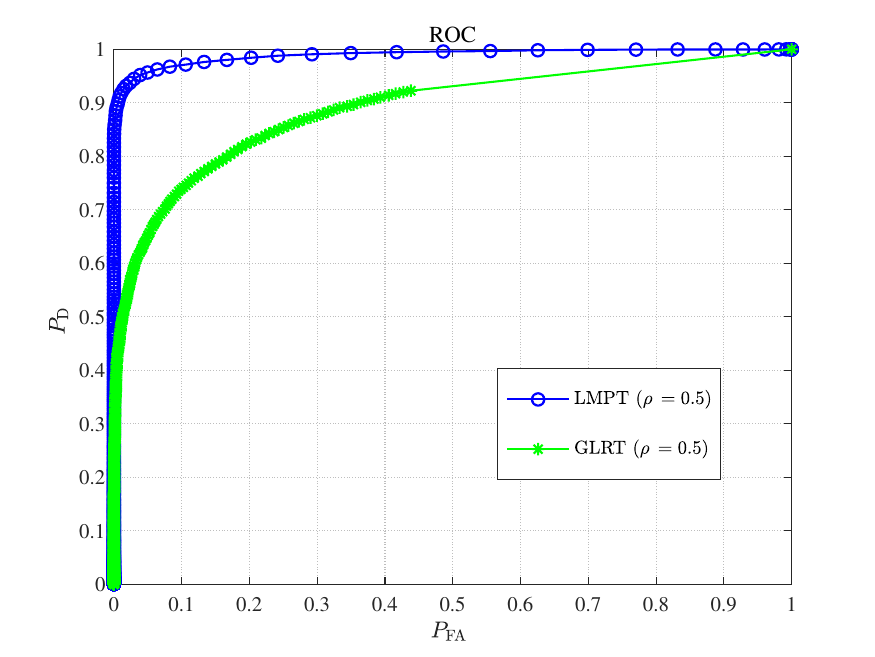}
		\label{fig4_2}}
	\hspace{-0.7cm}	
	\subfloat[Channel correlation $\rho = 1$]{
		\includegraphics[scale=0.45]{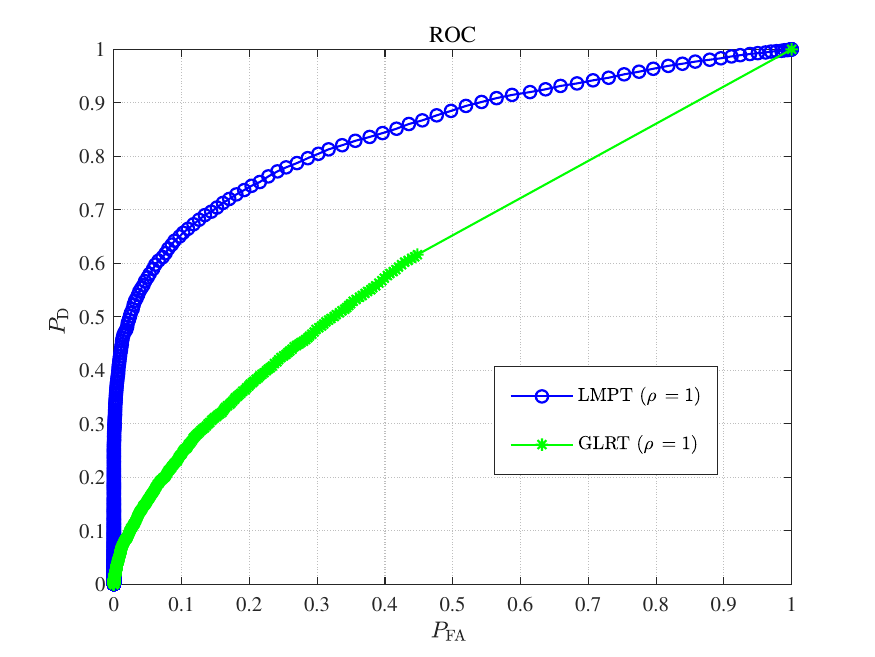}
		\label{fig4_3}}
	\caption{Comparison of ROC between the jamming detection schemes based on the LMPT and modified GLRT, where ${\text{JNR}} =2\,{\text{dB}}$ and $\tau = 5$.}
	\label{fig4}
\end{figure*}

\begin{figure*}[!t]
	\vspace{-1.2cm}
	\hspace{-0.5cm}
	\subfloat[MSE of the estimated inner-product's norm]{
		\includegraphics[scale=0.59]{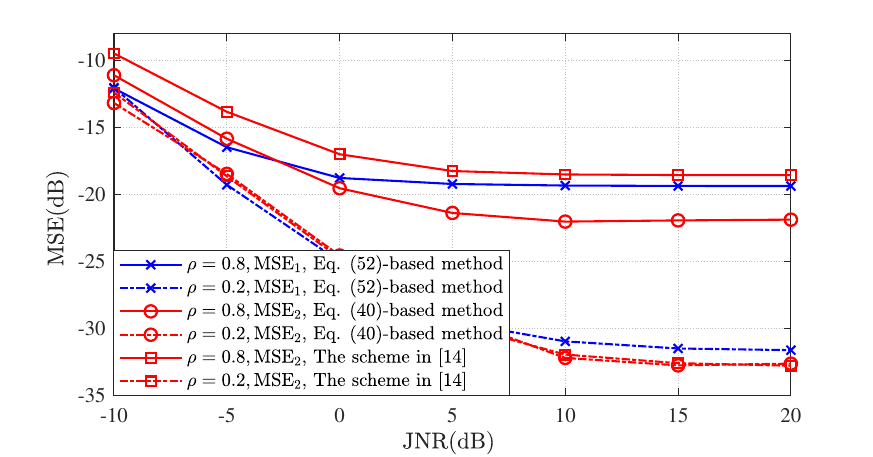}
		\label{fig6_1}}
	\hspace{0.06cm}
	\centering	
	\subfloat[CoS of the estimated inner-product's phase difference]{
		\includegraphics[scale=0.59]{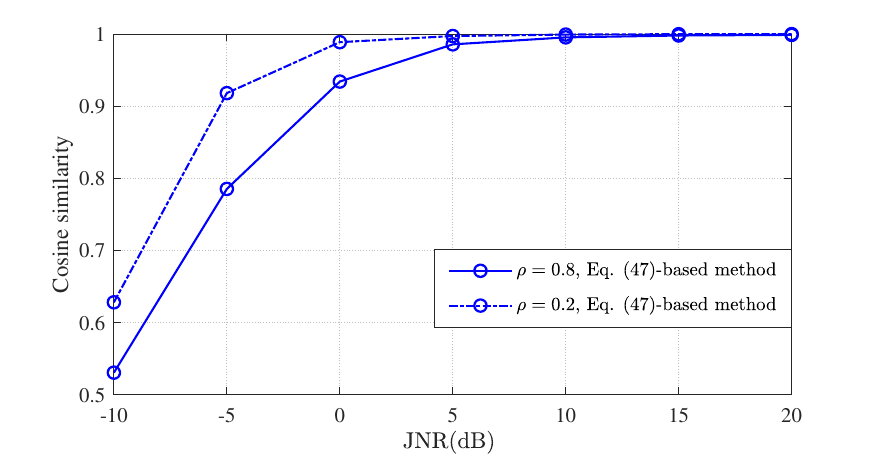}
		\label{fig6_2}}
	\caption{The performance of the proposed estimation scheme for the jamming pilot's inner-product.}
	\label{fig6}
	\vspace{-0.5cm}
\end{figure*}

\begin{figure*}[!t]
	\begin{minipage}[b]{0.5\textwidth}
		\includegraphics[scale=0.58]{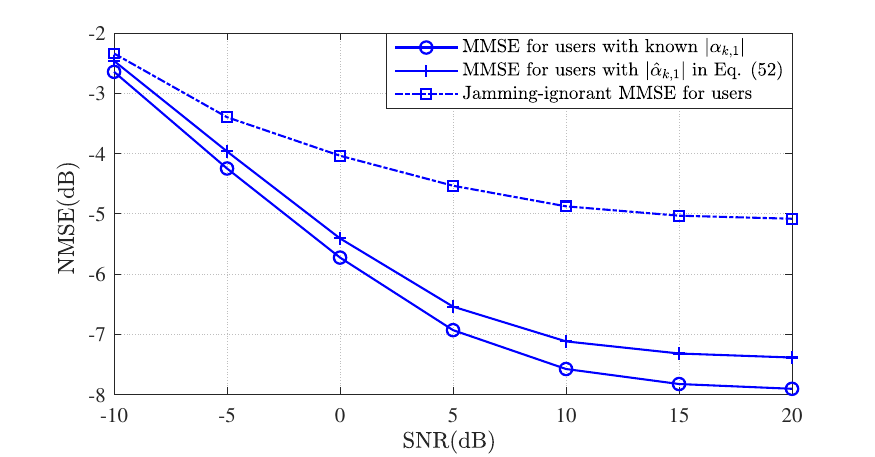}
		\caption{NMSE of the user channel estimation with known and estimated inner-product, where ${\text{JSR}} =0\,{\text{dB}}$, $N_{\rm{BS}} = 48$, $N_{\rm{UE}} = 12$ and $\tau = 4$.}
		\label{fig7_2}
	\end{minipage}
	\vspace{-0.1cm}
	\begin{minipage}[b]{0.5\textwidth}
		\vspace{0.05cm}
		\includegraphics[scale=0.575]{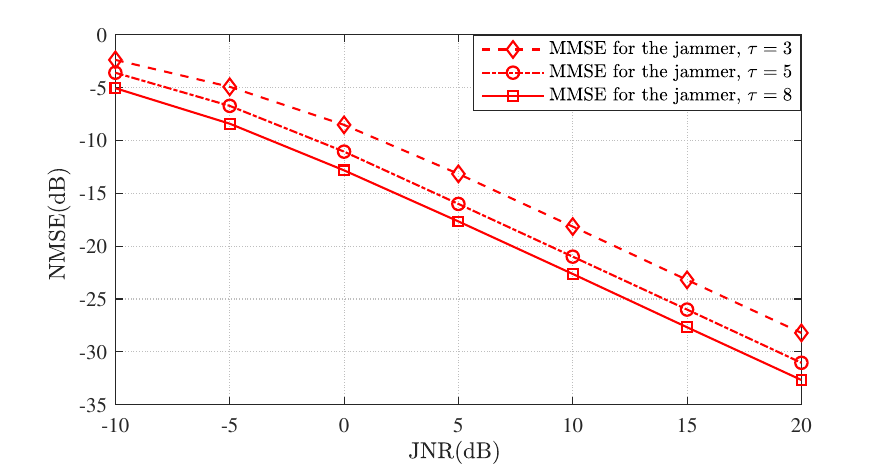}
		\caption{NMSE of the jamming channel estimation with known inner-product, where ${\text{JSR}} =0\,{\text{dB}}$, $N_{\rm{BS}} = 48$, $N_{\rm{UE}} = 12$ and $\tau = 3,5,8$.}
		\label{fig7_1}
	\end{minipage}
	\vspace{-0.5cm}
\end{figure*}

\vspace{0.5em}
\subsection{Performance of the Proposed Jamming Detection Schemes}
Fig. \ref{fig2} shows the simulation and theoretical results of the proposed LMPT-based jamming detection scheme for a system with $\rho = 0.9$, ${\text{JNR}} =0\,{\text{dB}}$, $N_{\rm{BS}}  = 4$, $N_{\rm{UE}} = 2$ and $\tau = 2$, $5$. 
The analytical expressions in Eqs. \eqref{PFA} and \eqref{PD} are in the forms  of infinite-sum-series.  In the simulation, a truncated  summation of the first $100$ terms is employed, along with a  normalization of $ a_{k,m}$’s or $\bar a_{k,m}$’s. 
The upper and middle figures show the curves of detection probability $P_{\rm{D}}$ and false-alarm probability $P_{\rm{FA}}$ versus $\gamma'_{\rm{LMPT}}$, respectively.
The upper two graphs demonstrate that the theoretical and simulation results match perfectly for different pilot lengths.
The lower figure shows  the $P_{\rm{D}}$ versus $P_{\rm{FA}}$ curve, also referred to as the receiver operating characteristic (ROC).
The curve with $\tau = 5$ rises faster in the low false-alarm region compared to the curve with $\tau = 2$, indicating superior performance. This phenomenon reveals that the increase of the pilot length and the joint utilization of multiple projected observation vectors can effectively improve the detection probability.
Fig. \ref{fig1} depicts the ROC cureve of the proposed LMPT-based detection scheme for various JNR values in the scenario in which $\tau = 2$, $\rho = 0.9$, $N_{\rm{BS}}  = 4$ and $N_{\rm{UE}} = 2$.
It can be seen that the simulation results are in good agreement  with the theoretical ones obtained in Theorem \ref{theorem2} under different JNR. 
The ROC curve rises with increasing JNR, indicating improved detection performance, which aligns with our expectations. The proposed LMPT-based detection  achieves 88.41\% detection probability with a false-alarm of 0.1 when ${\text{JNR}} =5\,{\text{dB}}$.

Fig. \ref{fig3} illustrates the ROC curves for different channel correlation levels $\rho$ for the LMPT-based  detection scheme under $\tau = 2$ and $\left( N_{\rm{BS}}, N_{\rm{UE}}\right) = (4,2)$.
It is observed that the detection performance decreases as the channel correlation $\rho$ increases.
Specifically, the change of $\rho$ from 0.2 to 0.8 at ${\text{JNR}} =0\,{\text{dB}}$, results in a 12.44\% decrease in $P_{\rm{D}}$ at $P_{\rm{FA}} = 0.1$.
Under strongly correlated channel conditions, the observations contain less information about the jammer, leading to worse detection performance.
A detection probability of 97.36\% at $0.01$ false-alarm can be achieved with ${\text{JNR}} =0\,{\text{dB}}$ and $\rho = 0.2$.
This indicates that as the channel approaches independent channel, the proposed scheme has high performance.



The GLRT detector under independent channel fading has been proposed in \cite{akhlaghpasand2017jamming}, and we use its modified form as a comparison scheme for the LMPT-based jamming detection. Let ${{\bf{Y}}}_{k}  = \left[{\bar{\bf{y}}}_{k,2},\ldots, {\bar{\bf{y}}}_{k,\tau}\right] $ be an ${N_{\rm{b}}} \times \left(\tau -1 \right) $ matrix, and denote the $n$-th column of ${{\bf{Y}}}^{\rm{H}}_{k}$ as ${\tilde{\bf{y}}}_{k,n} \in \mathbb{C}^{\left( \tau -1 \right) \times 1}$.
Then we can get the modified GLRT detector as
\begin{align}
	\vspace{-0.4cm}
	T_{\rm{GLRT}} =  {\rm{max}}\left( \frac{\sum_{n=1}^{N_{\rm{b}}}\left| {\tilde{\bf{y}}}_{k,n}^{\rm{H}} {\bf{1}}_{\tau -1}\right| }{{N_{\rm{b}}}\left(\tau-1 \right)^2 } - \frac{1}{\tau-1},0\right) \stackrel{\mathcal{H}_{1}}{\underset{\mathcal{H}_{0}}{\gtrless}}\gamma'_{\rm{GLRT}}.
	\vspace{-0.4cm}
\end{align}
A comparison of the ROC performance between the LMPT-based  detection scheme and the modified GLRT-based  scheme is presented in Fig. \ref{fig4}.
From Fig. \ref{fig4_1}, it can be observed that the LMPT-based scheme can achieve the same performance as the GLRT one when $\rho=0$.
Figs. \ref{fig4_2} and \ref{fig4_3} depict the detection performance 
for $\rho = 0.5$ and $\rho = 1$, respectively.
In both cases, the LMPT-based scheme outperforms the GLRT-based scheme as the LPMT-based detection scheme exploits the full eigen-space information. Specifically, under $\rho = 0.5$, the LMPT-based scheme can improve the detection probability  by 32.22\% at $P_{\rm{FA}} = 0.1$ compared to the GLRT scheme.

\subsection{Performance of the Proposed  Two-Step Estimation Scheme}
We first evaluate the performance of the estimation scheme of the jamming pilot's inner-product in the two-step estimation.
To measure the estimation accuracy of the inner-product's norm, we calculate the
MSE  of $\left| {{\hat\alpha _{k,1}}} \right |$ and $\left\lbrace \left| {{\hat\alpha _{k,i}}} \right |,i\in\vmathbb{t} \right\rbrace $ defined as
$
 	{\rm{MSE}}_1 = \frac{1}{K}\sum_{k\in\mathbb{K}}\left| | {{\hat\alpha _{k,1}}}| -  |   {{\alpha _{k,1}}}|\right| ^2, 
$
and
$
{\rm{MSE}}_{2} = \frac{1}{\left(\tau-1 \right) K}\sum_{k\in\mathbb{K}}\sum_{i\in\vmathbb{t}}\left| |{{\hat\alpha _{k,i}}}| - |{{\alpha _{k,i}}}| \right |^2.
$
The average MSEs of the inner-product's norm estimation scheme under different channel correlation levels, i.e., $\rho = 0.2,\,0.8$, are demonstrated in Fig. \ref{fig6_1}. The jamming-to-signal ratio (JSR) is $0\,{\rm{dB}}$, $N_{\rm{BS}} = M_{\rm{BS}}= 64$ and $\tau= 4$.
For the weighting coefficient in Eq. \eqref{alpha1 est total}, $\epsilon = 0.1$ is empirically chosen  for the simulation.
We adopt the norm estimation scheme based on the asymptotic property proposed in \cite{akhlaghpasand2019jamming} as the comparison scheme.
Under weakly correlated channel conditions, the performance using Eq. \eqref{40} is similar to that of the comparison scheme, while it has  better performance under strongly correlated channel conditions. Specifically, the average ${\rm{MSE}}_2$ reaches $-29.66\,{\rm{dB}}$ at $\rho=0.2$ and ${\rm{JNR}}=5\,{\rm{dB}}$. Under the same JNR conditions, the proposed estimate in Eq. \eqref{40} can improve the performance by $3.13\,{\rm{dB}}$ \textcolor{black}{compared to the scheme in \cite{akhlaghpasand2019jamming}} at $\rho=0.8$. 
In addition,
it can be seen that the performance trend of the average ${\rm{MSE}}_1$ is approximately the same as that of the average ${\rm{MSE}}_2$ for different levels of channel correlation.
As another part of the estimation performance of the inner-product, the cosine similarity (CoS) of the estimated phase difference  is defined as
$
	{\rm{CoS}} = \frac{1}{\left(\tau-2 \right)K }\sum_{k\in\mathbb{K}}\sum_{i\in\vmathbb{t}\setminus \left\lbrace 2\right\rbrace  }{\rm{cos}} \left(\hat{\theta}_{k,i} - {\theta}_{k,i} \right).
$
The CoS using the proposed  estimation in  Eq. \eqref{angle diff estimation} is illustrated in Fig. \ref{fig6_2}. It can be seen that at ${\rm{JNR}}=5\,{\rm{dB}}$, the average CoS can reach 98.73\% and 99.6\% under $\rho=0.2$ and $0.8$, respectively.
In conclusion, the proposed inner-product estimation scheme is robust under different channels and has better performance under independent channels.
On this basis, we further evaluate the performance of the proposed users' channel estimation with inner-product estimation. The normalized mean squared error (NMSE) of the users' channel estimations is defined as
$	{\rm{NMSE}}_{1} = \frac{1}{K}\sum_{k\in\mathbb{K}}\frac{\left \| {\hat{{\bf{h}}}_{k}} -{{{\bf{h}}}_{k}}  \right \|^2 }{\left \| {{{\bf{h}}}_{k}}  \right \|^2 }.$
Fig. \ref{fig7_2} shows the users' channel estimation performance with the known $\left| {{\alpha _{k,1}}} \right |$ and the  estimatied $\left| {{\hat \alpha _{k,1}}} \right |$ utilizing the scheme in Section \ref{estimation of pilot}. 
The users' channel estimation scheme aided by the proposed  estimation scheme of $\left| {{\alpha _{k,1}}} \right |$ realizes a performance close to that of the scheme with known $\left| {{\alpha _{k,1}}} \right |$.
The proposed users' channel estimation scheme can better capture the channel,
reducing NMSE by $2\,{\text{dB}}$ at ${\text{SNR}} =5\,{\text{dB}}$ compared to the jamming-ignorant MMSE channel estimation scheme.

Next, we analyze the effect of the pilot length on the performance of the proposed jamming channel estimation scheme. Note that the NMSE of the jamming channel is defined as ${\rm{NMSE}}_2 = \frac{1}{K}\sum_{k\in\mathbb{K}}\frac{\left \| {\hat{{\bf{h}}}_{k}^{{\rm{JM}}}} -{{{\bf{h}}}_{k}^{{\rm{JM}}}}  \right \|^2 }{\left \| {{{\bf{h}}}_{k}^{{\rm{JM}}}}  \right \|^2 }$.
For a system with ${\text{JSR}} =0\,{\text{dB}}$, $N_{\rm{BS}} = 48$, and $N_{\rm{UE}} = 12$,  the NMSE of the proposed channel estimation scheme with known inner-product versus the JNR for three pilot lengths is illustrated in Fig. \ref{fig7_1}. 
The increase in $\tau$ brings about a steady improvement in the channel estimation performance of the jammer. 

\begin{figure}[!htp]
	\includegraphics[scale=0.52]{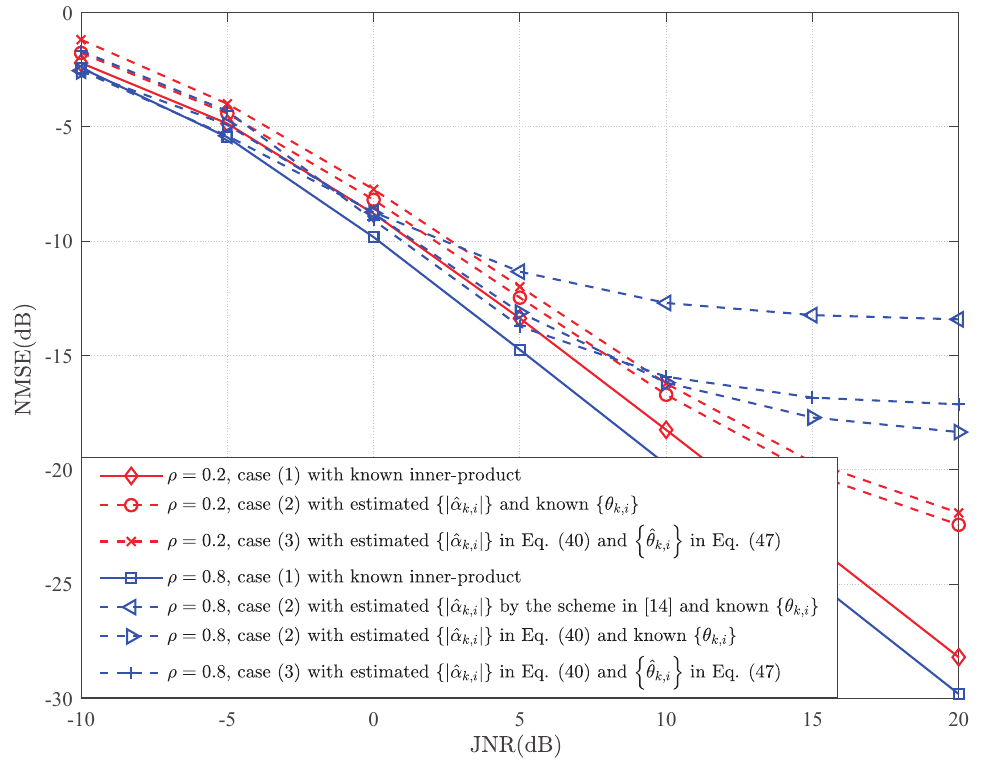}
	\caption{NMSE of the jamming channel  estimation scheme under different inner-product estimation strategies, where $N_{\rm{BS}} = 48$, $N_{\rm{UE}} = 12$ and $\tau = 4$.}
	\label{fig8}
	\vspace{-0.6cm}
\end{figure}
For the comparison of the jamming channel estimation performance with and without inner-product estimation, 
Fig. \ref{fig8} shows the NMSE for three cases, 
namely (1) with known inner-product $\bar{\boldsymbol{\alpha}}_{k}$, (2) with the estimated $\left\lbrace \left| {{\hat\alpha _{k,i}}} \right |,i\in\vmathbb{t} \right\rbrace $ and the known  $\left\lbrace {\theta}_{k,i},i\in\vmathbb{t}\setminus\left\lbrace  2\right\rbrace \right\rbrace $; (3) with the estimated ${\hat{\bar{{\boldsymbol{\alpha}}}}}_{k}$. 
For weakly correlated channels, i.e., $\rho = 0.2$, the performance is close in both cases (2) and (3). And, the difference between the performance in the above two cases and that  of the MMSE channel estimation with known inner-product at ${\text{JNR}} =10\,{\text{dB}}$ is $1.54\,{\text{dB}}$.
For strongly correlated channels, the performances under cases (2) and (3) are also close and outperform the estimation scheme in \cite{akhlaghpasand2019jamming}. This confirms the effectiveness of the proposed inner-product estimation strategy. Specifically, the multiple projected observation vectors based MMSE channel estimation  can achieve $-15.93\,{\text{dB}}$ NMSE at ${\text{JNR}} =10\,{\text{dB}}$ and $\rho = 0.8$, with the aid of the proposed  inner-product  estimation.

\begin{figure*}[hb]
	
	\newcounter{TempEqCnt5} 
	\setcounter{TempEqCnt5}{\value{equation}} 
	\setcounter{equation}{58} 
	\centering 
	\hrulefill 
	\vspace*{-2pt} 
	{\small{
			\begin{align} \label{four order}
				\mathbb{E}\left[|\alpha_{k,i}|^2|\alpha_{k,j}|^2\right] 
				& =  \underbrace{\mathbb{E}\left[ \left( \sum_{n=1}^{\tau}\left | {\phi}_{i,n}  \right |^2 \left |  {\psi}^{*}_{k,n} \right |^2\right)  \left( \sum_{n=1}^{\tau}\left | {\phi}_{j,n}  \right |^2 \left |  {\psi}^{*}_{k,n} \right |^2\right)  \right]}_{\Psi_1 } + \underbrace{\mathbb{E}\left[
					\left( \sum_{n=1}^{\tau}\left | {\phi}_{i,n}  \right |^2 \left |  {\psi}^{*}_{k,n} \right |^2\right) 
					\left( \sum_{n=1}^{\tau} \sum_{m \ne n} {\phi}^*_{j,n}{\phi}_{j,m}{\psi}_{k,n}{\psi}^{*}_{k,m}\right)
					\right] }_{\Psi_2 } \nonumber \\
				&\hspace{0.1cm} + \underbrace{\mathbb{E}\left[
					\left( \sum_{n=1}^{\tau} \sum_{m \ne n} {\phi}^*_{i,n}{\phi}_{i,m}{\psi}_{k,n}{\psi}^{*}_{k,m}\right) 
					\left( \sum_{n=1}^{\tau}\left | {\phi}_{j,n}  \right |^2 \left |  {\psi}^{*}_{k,n} \right |^2 \right)
					\right] }_{\Psi_3 } + \underbrace{\sum_{n=1}^{\tau} \sum_{m \ne n} {\phi}^*_{i,n}{\phi}_{i,m}{\phi}_{j,n}{\phi}^*_{j,m}\mathbb{E}\left[|{\psi}_{k,n}|^2\right] \mathbb{E}\left[|{\psi}^{*}_{k,m}|^2\right]}_{\Psi_4 }
				.
	\end{align} }}
	\setcounter{equation}{\value{TempEqCnt5}} 
	\vspace*{-1cm}
\end{figure*}

\vspace{-1em}
\section{Conclusion}
\label{Conclusion}


This paper investigated jamming detection and channel estimation in beamspace massive MIMO beam training under random active attacks by malicious jammers.
The hypothesis testing problem was formulated using multiple projected observation vectors  corresponding to the unused pilots obtained via pilot space projection.
The likelihood functions of these projected observation vectors were also derived.
With known second-order statistics of the channels, we proposed the jamming detection scheme based on LMPT.
We derived analytical expressions for the detection probability $P_{\rm{D}}$ and false alarm probability $P_{\rm{FA}}$ of the LMPT scheme.
For channel acquisition, we proposed a two-step MMSE-based channel estimation method. First, the unknown parameters introduced by the jammer's random pilots are estimated. Second, channel estimation is performed using the projected observation vectors and the estimated parameters.
Simulations verified our derivations and demonstrated how various system parameters, such as pilot length and channel correlation level, impact the performance of the proposed schemes.
The simulations demonstrated the superior performance of the proposed schemes compared to the baseline methods over highly correlated channels.


\appendix
\subsection{Proof of Theorem \ref{theorem1}}\label{Appendix}

Under ${\mathcal{{H}}}_0$, ${\bar{\bf{y}}}_{k,i }$ are i.i.d. $\mathcal{CN}({\bf{0}}, {{\bf{I}}}_{})$. We have from Eq. \eqref{LMPT detector}
\begin{align}
	T_{\rm{LMPT}}
	& = {\bar{\bf{y}}}_{k }^{\rm{H}}\left({\mathbf{I}}_{\tau-1} \otimes {\bf{V}}_k\right) \left({\mathbf{I}}_{\tau-1} \otimes {\boldsymbol{\Lambda}}_k\right) \left({\mathbf{I}}_{\tau-1} \otimes {\bf{V}}_k\right)^{\rm{H}}{\bar{\bf{y}}}_{k },\nonumber \\
	&=  {{{\bf{r}}}_{k }^{\rm{H}}\left({\mathbf{I}}_{\tau-1} \otimes {\boldsymbol{\Lambda}}_k\right){{\bf{r}}}_{k }}
	= \frac{1}{2}\sum_{n=1}^{\rho_k}  \lambda_{k,n}\sigma^2 x_{k,n} ,
\end{align}	
where ${{\bf{r}}}_{k } = [r_{k,1},\ldots,r_{k,\left(\tau - 1 \right) N_{\rm{b}}}]^{\rm{T}} = \left({\mathbf{I}}_{\tau-1} \otimes {\bf{V}}_k\right)^{\rm{H}}{\bar{\bf{y}}}_{k } $, and $x_{k,n}=\sum_{i:i\,{\rm{mod}}\,N_{\rm{b}}=n} \frac{|\sqrt{2}r_{k,i}|^2}{\sigma^2}$.
	It can be concluded that  ${{\bf{r}}}_{k } \sim\mathcal{CN}({\bf{0}}, \sigma^2{{\bf{I}}}_{})$ since
	$\left({\mathbf{I}}_{\tau-1} \otimes {\bf{V}}_k\right)$ is unitary and $ {\bar{\bf{y}}}_{k } \sim\mathcal{CN}({\bf{0}}, \sigma^2{{\bf{I}}}_{})$, and thus $\frac{|\sqrt{2}r_{k,i}|^2}{\sigma^2} \sim  \chi_2^2(0)$, $x_{k,n} \sim  \chi_{2(\tau-1)}^2(0)$.
We have shown that $T_{\rm{LMPT}}$ under $\mathcal{H}_0$ is distributed as a linear combination of independent chi-squared variables $x_{k,n}$ each having $2(\tau-1)$ degrees of freedom.

Under  $\mathcal{H}_1$, based on Eq. \eqref{LMPT detector}, the LMPT-based detector can be further represented as
\begin{align}
	T_{\rm{LMPT}}
	&= \sum_{i=2}^{\tau}\sum_{n=1}^{N_{\rm{b}}}
	{\lambda_{k,n}|{\bar{\bf{y}}}^{\rm{H}}_{k,i}  {{\bf{v}}}_{k,n} |^2}
	=\sum_{i=2}^{\tau}{{\bar{\bf{y}}}_{k,i }^{\rm{H}}{\tilde{\bf{R}}}_{{\rm{JM}},k}{\bar{\bf{y}}}_{k,i }}.
\end{align}
By converting the summation of the quadratic forms of length-$N_{\rm{b}}$ projected observation vectors $ \left\lbrace {\bar{\bf{y}}}_{k,i},i\in\vmathbb{t} \right\rbrace $ to the quadratic form of $  {\bar{\bf{y}}}_{k}$, the above equation can be transformed into
\begin{align} \label{proof a lmpt}
	T_{\rm{LMPT}}
	&=  {\bar{\bf{y}}}_{k }^{\rm{H}} \left({\mathbf{I}}_{\tau-1} \otimes {\tilde{\bf{R}}}_{{\rm{JM}},k}\right) {\bar{\bf{y}}}_{k} = {{\bar{\bf{r}}}_{k }^{\rm{H}}{\bf{B}}_k
		{\bar{\bf{r}}}_{k }},
\end{align}
where $ {\bar{\bf{r}}}_{k } = {{\bf{R}}}^{-1/2}_{{\bar{\bf{y}}}_{k }} {\bar{\bf{y}}}_{k}\sim\mathcal{CN}({\bf{0}}, {{\bf{I}}}_{})$ due to ${\bar{\bf{y}}}_{k}\sim\mathcal{CN}({\bf{0}}, {{\bf{R}}}_{{\bar{\bf{y}}}_{k }})$  and ${\bf{B}}_k $ is defined in Eq. \eqref{BK}.
Using the EVD  ${\bf{B}}_k  = \mathbf{	Q}^{\rm{H}}_k {\bf{D}}_k \mathbf{Q}_k$, where the $(i,i)$-th element of $ \mathbf{D}_k $ is $\epsilon_{k,i}$, $ T_{\rm{LMPT}}$ in Eq. \eqref{proof a lmpt} under ${\mathcal{{H}}}_1$ can be transformed into
\begin{align}
	T_{\rm{LMPT}}
	& = {\bar{\bf{r}}}^{\rm{H}}_{k } \mathbf{	Q}^{\rm{H}}_k {\bf{D}}_k \mathbf{	Q}_k{\bar{\bf{r}}}_{k } \xlongequal[]{{\bar{\bar{\bf{r}}}}_{k }=\mathbf{Q}_k{\bar{\bf{r}}}_{k }}
	 \sum_{n=1}^{\varphi_k}\frac{\epsilon_{k,n}}2\left|\sqrt{2}{\bar{\bar{{r}}}}_{k,n}\right|^2,  
\end{align}	
	where ${\bf{D}}_k$ and $\mathbf{Q}_k$ are the eigenvalue and eigenvector matrices.
  ${\bar{\bar{{r}}}}_{k,n}$ is the $n$-th element of ${\bar{\bar{\bf{r}}}}_{k }$ and ${\bar{\bar{{r}}}}_{k,n}\sim\mathcal{CN}(0, 1)$. It's easy to see that $T_{\rm{LMPT}}$ under ${\mathcal{H}}_1$ is a linear
combination of $\varphi_k$ independent chi-squared variables each having $2$ degrees of freedom. 
From the result in \cite[Eq. (2.4)]{provost1996exact}, the PDFs of $T_{\rm{LMPT}}$ under $\mathcal{H}_0$ and $\mathcal{H}_1$ can be written as Eqs. \eqref{PDF h0} and \eqref{PDF h1}.

\subsection{Proof of Lemma \ref{lemma1}}\label{Appendix B}

Recall that $	\alpha_{k,i} = {\boldsymbol{\phi}}^{\rm{T}}_{i}{\boldsymbol{ \psi}}^{*}_{k}
=\sum_{n=1}^{\tau} {\phi}_{i,n} {\psi}^{*}_{k,n}$, and ${\psi}_{k,n} \sim \mathcal{CN}\left(0,\frac{1}{\tau} \right) $. Thus, we get $\mathbb{E}\left[ |{\psi}_{k,n}|^2\right] =\frac{1}{\tau}$ and $\mathbb{E}\left[ |{\psi}_{k,n}|^4\right] =\frac{2}{\tau^2}$.
Then, the $(i-1)$-th element of $\mathbb{E}\left[ {\bf{x}}_{k} \right]$ can be obtained as
{{
\begin{align}
	&\mathbb{E}\left[|\alpha_{k,i}|^2\right]\nonumber\\
	&=\mathbb{E}\left[\sum_{n=1}^{\tau}\left | {\phi}_{i,n}  \right |^2 \left |  {\psi}^{*}_{k,n} \right |^2
	+\sum_{n=1}^{\tau} \sum_{m \ne n} {\phi}^*_{i,n}{\phi}_{i,m}{\psi}_{k,n}{\psi}^{*}_{k,m}
	\right].
\end{align} }}
Therefore, we can obtain the expectation of $ {\bf{x}}_{k} $ as Eq. \eqref{expectation of xk}.

Next, we analyze $\mathbb{E}\left[ {\bf{x}}_{k}{\bf{x}}^{\rm{T}}_k \right]$. For the $(i-1, j-1)$-th element of $\mathbb{E}\left[ {\bf{x}}_{k}{\bf{x}}^{\rm{T}}_k \right]$, we have Eq. \eqref{four order} as shown at the bottom of the previous page.
For the first term, we have
\setcounter{equation}{59}
\begin{align}
	\Psi_1
	=\,&\sum_{n=1}^{\tau}\left | {\phi}_{i,n}  \right |^2\left | {\phi}_{j,n}  \right |^2\mathbb{E}\left[ \left |  {\psi}^{*}_{k,n} \right |^4
	\right] \nonumber \\
	&+ \sum_{n=1}^{\tau}\sum_{m\ne n} \left| {\phi}_{i,n}  \right |^2\left | {\phi}_{j,m}  \right |^2 \mathbb{E}\left[\left |  {\psi}^{*}_{k,n} \right |^2\right]  \mathbb{E}\left[\left |  {\psi}^{*}_{k,m} \right |^2  \right] 
	\nonumber \\
	=\,&\frac{2}{\tau^2}\sum_{n=1}^{\tau}\left | {\phi}_{i,n}  \right |^2\left | {\phi}_{j,n}  \right |^2 + \frac{1}{\tau^2}\sum_{n=1}^{\tau}\sum_{m\ne n} \left| {\phi}_{i,n}  \right |^2\left | {\phi}_{j,m}  \right |^2.
	\nonumber \\
	=\,&\frac{1}{\tau^2}\sum_{n=1}^{\tau}\left | {\phi}_{i,n}  \right |^2\left | {\phi}_{j,n}  \right |^2 + \frac{1}{\tau^2}.
\end{align}
The second term can be expressed as
\begin{align}
	\Psi_2&=\sum_{n=1}^{\tau} \sum_{m \ne n} \left | {\phi}_{i,n}  \right |^2{\phi}^*_{j,n}{\phi}_{j,m} \mathbb{E}\left[\left |  {\psi}^{*}_{k,n} \right |^2{\psi}_{k,n}\right] \mathbb{E}\left[{\psi}^{*}_{k,m}\right] \nonumber \\
	&\hspace{0cm} 
	+ \sum_{t=1}^{\tau} \sum_{n \ne t} \sum_{m \ne n} \left | {\phi}_{i,t}  \right |^2{\phi}^*_{j,n}{\phi}_{j,m} \mathbb{E}\left[\left |  {\psi}^{*}_{k,t} \right |^2\right]\mathbb{E}\left[ {\psi}_{k,n}\right] \mathbb{E}\left[{\psi}^{*}_{k,m}\right] 
	\nonumber \\
	&= 0.
\end{align}
Similarly, $\Psi_3 = 0$. Finally, for the fourth term, we get the following expression as
\begin{align}
	\Psi_4
	=\frac{1}{\tau^2}\sum_{n=1}^{\tau} \sum_{m \ne n} {\phi}^*_{i,n}{\phi}_{i,m}{\phi}_{j,n}{\phi}^*_{j,m}.
\end{align}
Taking the above analysis together, $\mathbb{E}\left[|\alpha_{k,i}|^2|\alpha_{k,j}|^2\right]$ can be expressed as
\begin{align}
	\mathbb{E}\left[|\alpha_{k,i}|^2|\alpha_{k,j}|^2\right] 
	& = \frac{1}{\tau^2}\left({\rm{tr}}\left({\boldsymbol{\phi}}_{i}{\boldsymbol{\phi}}^{\rm{H}}_{i}{\boldsymbol{\phi}}_{j}{\boldsymbol{\phi}}^{\rm{H}}_{j} \right)  +1\right) .
\end{align}
Up to this point, we can obtain the expression for  $\mathbb{E}\left[ {\bf{x}}_{k}{\bf{x}}^{\rm{T}}_k \right]$ as shown in Eq. \eqref{covariace of xk}.

\subsection{Proof of Lemma \ref{lemma2}}\label{Appendix C}
For the expectation of the noise term $\tilde{\bf{v}}_{k}$, we can easily get Eq. \eqref{expectation of vk}  from the definition of $\tilde{\bf{v}}_{k}$.
Further, we analyze $\mathbb{E}\left[\tilde{\bf{v}}_{k} \tilde{\bf{v}}^{\rm{T}}_{k} \right]$. First, define ${\bar{\tilde{\bf{v}}}}_{k} = [{\tilde{{v}}}_{k,2},\ldots,{\tilde{{v}}}_{k,\tau}]^{\rm{T}}$, $\tilde{\bf{v}}_{k} = [{\tilde{{v}}}_{k,1},{\bar{\tilde{\bf{v}}}}_{k}]^{\rm{T}}$, then we have
\begin{align}
	\mathbb{E}\left[\tilde{\bf{v}}_{k} \tilde{\bf{v}}_{k}^{\rm{T}} \right] = \begin{pmatrix}
		\mathbb{E}\left[{\tilde{{v}}}_{k,1} {\tilde{{v}}}_{k,1} \right]  &
		\mathbb{E}\left[{\tilde{{v}}}_{k,1} {\bar{\tilde{\bf{v}}}}_{k}^{\rm{T}} \right] \\
		\mathbb{E}\left[{\bar{\tilde{\bf{v}}}}_{k} {\tilde{{v}}}_{k,1} \right] 
		&\mathbb{E}\left[{\bar{\tilde{\bf{v}}}}_{k} {\bar{\tilde{\bf{v}}}}_{k}^{\rm{T}} \right] 
	\end{pmatrix}.
\end{align}

For the $(1,1)$-th block of the matrix $\mathbb{E}\left[\tilde{\bf{v}}_{k} \tilde{\bf{v}}_{k} \right]$, we have
{\small{
\begin{align} \label{noise covariance 1st item}
	&\mathbb{E}\left[{\tilde{{v}}}_{k,1} {\tilde{{v}}}_{k,1} \right] = \Phi_k  \nonumber \\
	& + \tau^2 {p_{t}^2}\mathbb{E}\left[
	\left| {\tilde{\bf{h}}}^{\rm{H}}_k{\tilde{\bf{h}}}_k\right| ^2
	\right]
	+
	2\tau {p_{t}}\mathbb{E}\left[{\tilde{\bf{h}}}^{\rm{H}}_k{\tilde{\bf{h}}}_k
	\Delta_{k,1} \right]  + \mathbb{E}\left[\left| \mathbf{n}^{\rm{H}}_{k,1}\mathbf{n}_{k,1}\right| ^2 \right]
	\nonumber
	\\&+ 4\tau^{3/2} {p^{3/2}_{t}}\mathbb{E}\left[{\tilde{\bf{h}}}^{\rm{H}}_k{\tilde{\bf{h}}}_k
	\Re \left[ \mathbf{n}^{\rm{H}}_{k,1} {\tilde{\bf{h}}}_k  \right]  \right]
	+2\tau {p_{t}}\mathbb{E}\left[{\tilde{\bf{h}}}^{\rm{H}}_k{\tilde{\bf{h}}}_k\right] 
	\mathbb{E}\left[\mathbf{n}^{\rm{H}}_{k,1}\mathbf{n}_{k,1}\right]
	\nonumber
	\\&+  \mathbb{E}\left[ \Delta^2_{k,1}   \right]
	+ 4\sqrt{\tau p_t} \mathbb{E}\left[  \Delta_{k,1}  \Re \left[ \mathbf{n}^{\rm{H}}_{k,1} {\tilde{\bf{h}}}_k \right] \right]
	+ 2 \mathbb{E}\left[ \Delta_{k,1} \mathbf{n}^{\rm{H}}_{k,1}\mathbf{n}_{k,1}  \right] \nonumber
	\\&+ 4\tau p_t\mathbb{E}\left[ \left(  \Re \left[ \mathbf{n}^{\rm{H}}_{k,1} {\tilde{\bf{h}}}_k   \right] \right)^2  \right]
	+4\sqrt{\tau p_t} \mathbb{E}\left[ \Re \left[ \mathbf{n}^{\rm{H}}_{k,1} {\tilde{\bf{h}}}_k \right] \mathbf{n}^{\rm{H}}_{k,1}\mathbf{n}_{k,1}  \right] .
\end{align}}}where $\Phi_k$
 denotes the expectation of all terms containing $\alpha_{k,1}$.
Based on the independence of ${\tilde{\bf{h}}}_k$, ${\tilde{\bf{h}}}^{\mathrm{JM}}_{k}$ and $\mathbf{n}_{k,1}$, we get
{\small{
\begin{align}
	\mathbb{E}\left[{\tilde{\bf{h}}}^{\rm{H}}_k{\tilde{\bf{h}}}_k
	\Delta_{k,1} \right]  = 0,\,
	\mathbb{E}\left[  \Delta_{k,1}  \Re \left[ \mathbf{n}^{\rm{H}}_{k,1} {\tilde{\bf{h}}}_k \right] \right] = 0,\,
	\mathbb{E}\left[ \Delta_{k,1} \mathbf{n}^{\rm{H}}_{k,1}\mathbf{n}_{k,1}  \right] = 0. 
\end{align}}}And for some cross terms, we can get
\begin{align}
	&\mathbb{E}\left[{\tilde{\bf{h}}}^{\rm{H}}_k{\tilde{\bf{h}}}_k
	\Re \left[ \mathbf{n}^{\rm{H}}_{k,1} {\tilde{\bf{h}}}_k  \right]  \right] = 0, \\
	&\mathbb{E}\left[ \Re \left[ \mathbf{n}^{\rm{H}}_{k,1} {\tilde{\bf{h}}}_k \right] \mathbf{n}^{\rm{H}}_{k,1}\mathbf{n}_{k,1}  \right]  = 0.
\end{align}
In order to solve for some particular terms, we need to use the
following lemma.

\begin{lemma} \label{lemma4}
	Let ${\tilde{\bf{h}}}$ be an $M \times 1$ complex
	Gaussian vector, and ${\tilde{\bf{h}}}\sim \mathcal{CN}\left( {\bf{0}}, {\tilde{\bf{R}}}\right) $. Then 
	$
		\mathbb{E}\left[
		\left| {\tilde{\bf{h}}}^{\rm{H}}{\tilde{\bf{h}}}\right| ^2
		\right] =  \left|{\rm{tr}}\left({\tilde{\bf{R}}}  \right)  \right| ^2 +\left\|{\tilde{\bf{R}}} \right\|_{\rm{F}}^2  .
	$
\end{lemma}
\begin{proof}
	Let ${{\bf{h}}} = {\tilde{\bf{R}}}^{-1/2}{\tilde{\bf{h}}}$, then ${{\bf{h}}} \sim \mathcal{CN}\left({\bf{0}},{\bf{I}}_{M} \right) $ and
	\begin{align}
		\vspace{-1cm}
		\left| {\tilde{\bf{h}}}^{\rm{H}} {\tilde{\bf{h}}} \right| ^2
		= 	\left| {{\bf{h}}}^{\rm{H}}{\tilde{\bf{R}}} {{\bf{h}}} \right| ^2
		  &=  \sum_{i,j,i',j'=1}^{M}r_{i,j}r^*_{i',j'}h_i^{*}h_jh_{i'}h^*_{j'},
	\end{align}
where $h_i$ is the $i$-th element of $ {{\bf{h}}}$ and $r_{i,j}$ is the $(i,j)$-th element of ${\tilde{\bf{R}}}$. From $\mathbb{E}\left[h_i \right] = 0 $, $\mathbb{E}\left[h_i^* h_j \right] = \delta_{i,j}$ and $\mathbb{E}\left[h_i h_j \right] = 0$, we have
\begin{align} \label{70}
	\mathbb{E}\left[ \left| {\tilde{\bf{h}}}^{\rm{H}} {\tilde{\bf{h}}} \right| ^2\right] 
	&= 	\mathbb{E}\left[ \sum_{i=1}^{M}\left| r_{i,i}\right| ^2\left| h_i\right| ^4 \right] + \mathbb{E}\left[ \sum_{i\ne j} \left| r_{i,j}\right| ^2 \left| h_i\right| ^2\left| h_j\right| ^2 \right] \nonumber \\
	&+ \mathbb{E}\left[ \sum_{i\ne i'} r_{i,i}r^*_{i',i'} \left| h_i\right| ^2\left| h_{i'}\right| ^2\right],
\end{align}
where the three terms on the right-hand-side of Eq. \eqref{70} corresponding to the cases: $i=i' =j=j'$, $i=i'\ne j =j'$, and $i=j \ne i' =j'$. 
Since $	\mathbb{E}\left[ \left| h_i\right| ^4 \right] =2$ and $	\mathbb{E}\left[ \left| h_i\right| ^2 \right] = 1$, we have
\begin{align} 
	\mathbb{E}\left[ \left| {\tilde{\bf{h}}}^{\rm{H}} {\tilde{\bf{h}}} \right| ^2\right] 
	&= 	2 \sum_{i}\left| r_{i,i}\right| ^2  +  \sum_{i\ne j} \left| r_{i,j}\right| ^2 +  \sum_{i\ne i'} r_{i,i}r^*_{i',i'} ,\nonumber \\
	& = \left|\sum_{i} r_{i,i}\right| ^2 +  \sum_{i, j} \left| r_{i,j}\right| ^2 , \nonumber \\
	& = \left|{\rm{tr}}\left({\tilde{\bf{R}}}  \right)  \right| ^2 +\left\|{\tilde{\bf{R}}} \right\|_{\rm{F}}^2 .
\end{align}
\end{proof}

Based on Lemma \ref{lemma4}, we can get
\begin{align}
	\mathbb{E}\left[
	\left| {\tilde{\bf{h}}}^{\rm{H}}_k{\tilde{\bf{h}}}_k\right| ^2
	\right] =    \Theta_{k},\,\mathbb{E}\left[ \Delta^2_{k,1}   \right] = 2q_k^2\left\|{\tilde{\bf{R}}}_{{\rm{JM}},k} \right\|_{\rm{F}}^2.
\end{align}
Further, $\mathbb{E}\left[ \left(  \Re \left[ \mathbf{n}^{\rm{H}}_{k,1} {\tilde{\bf{h}}}_k   \right] \right)^2  \right]$ can be represented as
\begin{align} \label{73}
	&\mathbb{E}\left[ \left(  \Re \left[ \mathbf{n}^{\rm{H}}_{k,1} {\tilde{\bf{h}}}_k   \right] \right)^2  \right] \nonumber \\ &= \frac{1}{2}
	\left( \Re \left[  \mathbb{E}\left[ \mathbf{n}^{\rm{H}}_{k,1} {\tilde{\bf{h}}}_k  \mathbf{n}^{\rm{H}}_{k,1} {\tilde{\bf{h}}}_k   \right]\right]  +   \Re \left[\mathbb{E}\left[ \mathbf{n}^{\rm{H}}_{k,1} {\tilde{\bf{h}}}_k  \left( \mathbf{n}^{\rm{H}}_{k,1} {\tilde{\bf{h}}}_k\right) ^*   \right] \right]  \right) ,\nonumber \\
	& =  \frac{1}{2}\left( 
	\Re \left[ \sum_{i} \mathbb{E}\left[ n_{k,1,i}^*n_{k,1,i}^*\right] 
	\mathbb{E}\left[ {\tilde h}_{k,i}^*{\tilde h}_{k,i}^*\right]   \right] 
	\right. \nonumber \\
	&\hspace{2cm} \left. 
	+ \Re \left[ \sum_{i} \mathbb{E}\left[ \left| n_{k,1,i}\right| ^2\right] 
	\mathbb{E}\left[ \left| {\tilde h}_{k,i}\right| ^2\right]   \right] 
	\right) , \nonumber  \\
	& = \frac{1}{2}\sigma^2{\rm{tr}}\left( {\tilde{\bf{R}}}_{k}\right)  
	,
\end{align}
where ${\tilde{h}}_{k,i}$ and ${n}_{k,1,i}$ denote the $i$-th element of ${\tilde{\bf{h}}}_k$ and $\mathbf{n}_{k,1}$, respectively. Similarly to the derivation of Eq. \eqref{73}, we can obtain
\begin{align}
	\Phi_k  = {2\tau p_t q_k}\left| {\rm{tr}}\left( {\tilde{\bf{R}}}_{{\rm{JM}},k} {\tilde{\bf{R}}}_{k} \right)\right| + {2q_k\sigma^2}{\rm{tr}}\left( {\tilde{\bf{R}}}_{{\rm{JM}},k}  \right).
\end{align}
Finally, for the noise term
$\mathbb{E}\left[\left| \mathbf{n}^{\rm{H}}_{k,1}\mathbf{n}_{k,1}\right| ^2 \right]$, we have
{\small
\begin{align}
	&\mathbb{E}\left[\left| \mathbf{n}^{\rm{H}}_{k,1}\mathbf{n}_{k,1}\right| ^2 \right]= \sum_{i=1}^{N_{\rm{b}}} \mathbb{E}\left[ |n_{k,1,i}|^4\right] 
	+ \sum_{i=1}^{N_{\rm{b}}}\sum_{j\ne i} \mathbb{E}\left[ |n_{k,1,i}|^2\right]\mathbb{E}\left[|{n}_{k,1,j}|^2\right] ,\nonumber \\
	& = 2N_{\rm{b}}\sigma^4 + \left( N_{\rm{b}} - 1\right) N_{\rm{b}}\sigma^4 =  \left( N_{\rm{b}} + 1\right) N_{\rm{b}}\sigma^4.
\end{align}}Combining the above results we can obtain the $(1,1)$-th block of the matrix $\mathbb{E}\left[\tilde{\bf{v}}_{k} \tilde{\bf{v}}_{k} \right]$ as
$
	\mathbb{E}\left[{\tilde{{v}}}_{k,1} {\tilde{{v}}}_{k,1} \right] =\varrho_{1,1}
$.

For the $(i-1)$-th element of the $(1,2)$-th block vector $\mathbb{E}\left[{\tilde{{v}}}_{k,1} {\bar{\tilde{\bf{v}}}}_{k}^{\rm{T}} \right]$, we have
\begin{align}
	\mathbb{E}\left[{\tilde{{v}}}_{k,1} {\tilde{{v}}}_{k,i} \right] 
	=& \,\tau {p_{t}}\mathbb{E}\left[{\tilde{\bf{h}}}^{\rm{H}}_k{\tilde{\bf{h}}}_k\right] 
	\mathbb{E}\left[\mathbf{n}^{\rm{H}}_{k,i}\mathbf{n}_{k,i}\right]
	+ \mathbb{E}\left[\mathbf{n}^{\rm{H}}_{k,1}\mathbf{n}_{k,1}\right]\nonumber \\
	&\times \mathbb{E}\left[\mathbf{n}^{\rm{H}}_{k,i}\mathbf{n}_{k,i}\right] +\mathbb{E}\left[ \Delta_{k,1}  \Delta_{k,i}  \right]= \varrho_{1,2}.
\end{align}

The $(i-1, j-1)$-th element of the $(2,2)$-th block matrix $\mathbb{E}\left[{\bar{\tilde{\bf{v}}}}_{k} {\bar{\tilde{\bf{v}}}}_{k}^{\rm{T}} \right]$ can be denoted as
\begin{align}
	\mathbb{E}\left[{\tilde{{v}}}_{k,i} {\tilde{{v}}}_{k,j} \right]
	&= \mathbb{E}\left[\mathbf{n}^{\rm{H}}_{k,i}\mathbf{n}_{k,i}\right]\mathbb{E}\left[\mathbf{n}^{\rm{H}}_{k,j}\mathbf{n}_{k,j}\right]
	+\mathbb{E}\left[ \Delta_{k,i}  \Delta_{k,j}  \right],\nonumber \\
	& = \varrho_{2,2}.
\end{align}
Based on the above analysis, $\mathbb{E}\left[\tilde{\bf{v}}_{k} \tilde{\bf{v}}^{\rm{T}}_{k} \right]$ can be obtained.

\ifCLASSOPTIONcaptionsoff
  \newpage
\fi

\bibliographystyle{IEEEtran}
\bibliography{ReferencesNewAbbr}

\begin{thebibliography}{10}
\providecommand{\url}[1]{#1}
\csname url@samestyle\endcsname
\providecommand{\newblock}{\relax}
\providecommand{\bibinfo}[2]{#2}
\providecommand{\BIBentrySTDinterwordspacing}{\spaceskip=0pt\relax}
\providecommand{\BIBentryALTinterwordstretchfactor}{4}
\providecommand{\BIBentryALTinterwordspacing}{\spaceskip=\fontdimen2\font plus
\BIBentryALTinterwordstretchfactor\fontdimen3\font minus
  \fontdimen4\font\relax}
\providecommand{\BIBforeignlanguage}[2]{{%
\expandafter\ifx\csname l@#1\endcsname\relax
\typeout{** WARNING: IEEEtran.bst: No hyphenation pattern has been}%
\typeout{** loaded for the language `#1'. Using the pattern for}%
\typeout{** the default language instead.}%
\else
\language=\csname l@#1\endcsname
\fi
#2}}
\providecommand{\BIBdecl}{\relax}
\BIBdecl

\bibitem{wu2023simultaneous}
K.~Wu, J.~A. Zhang, X.~Huang, Y.~J. Guo, and L.~Hanzo, ``{Simultaneous beam and
  user selection for the beamspace mmWave/THz massive MIMO downlink},''
  \emph{IEEE Trans. Commun.}, vol.~71, no.~3, pp. 1785--1797, Mar. 2023.

\bibitem{hoang2018cell}
T.~M. Hoang, H.~Q. Ngo, T.~Q. Duong, H.~D. Tuan, and A.~Marshall, ``{Cell-free
  massive MIMO networks: Optimal power control against active eavesdropping},''
  \emph{IEEE Trans. Commun.}, vol.~66, no.~10, pp. 4724--4737, Oct. 2018.

\bibitem{kapetanovic2013detection}
D.~Kapetanovi{\'c}, G.~Zheng, K.-K. Wong, and B.~Ottersten, ``{Detection of
  pilot contamination attack using random training and massive MIMO},'' in
  \emph{Proc. IEEE Int. Symp. Person Indoor Mobile Radio Commun. (PIMRC)},
  London, United kingdom, Sep. 2013, pp. 13--18.

\bibitem{tugnait2015self}
J.~K. Tugnait, ``{Self-contamination for detection of pilot contamination
  attack in multiple antenna systems},'' \emph{IEEE Wirel. Commun. Lett.},
  vol.~4, no.~5, pp. 525--528, Oct. 2015.

\bibitem{vinogradova2016detection}
J.~Vinogradova, E.~Bj{\"o}rnson, and E.~G. Larsson, ``{Detection and mitigation
  of jamming attacks in massive MIMO systems using random matrix theory},'' in
  \emph{Proc. IEEE Workshop Signal Process. Adv. Wirel. Commun. (SPAWC)},
  Edinburgh, United kingdom, Jul. 2016, pp. 1--5.

\bibitem{akhlaghpasand2017jamming}
H.~Akhlaghpasand, S.~M. Razavizadeh, E.~Bj{\"o}rnson, and T.~T. Do, ``{Jamming
  detection in massive MIMO systems},'' \emph{IEEE Wirel. Commun. Lett.},
  vol.~7, no.~2, pp. 242--245, Apr. 2017.

\bibitem{qi2024anti}
X.~Qi, M.~Peng, H.~Zhang, and X.~Kong, ``{Anti-jamming hybrid beamforming
  design for millimeter-wave massive MIMO systems},'' \emph{IEEE Trans. Wirel.
  Commun.}, vol.~23, no.~8, pp. 9160--9172, Aug. 2024.

\bibitem{pirayesh2022jamming}
H.~Pirayesh and H.~Zeng, ``{Jamming attacks and anti-jamming strategies in
  wireless networks: A comprehensive survey},'' \emph{IEEE Commun. Surv.
  Tutor.}, vol.~24, no.~2, pp. 767--809, Mar. 2022.

\bibitem{xiong2015energy}
Q.~Xiong, Y.-C. Liang, K.~H. Li, and Y.~Gong, ``{An energy-ratio-based approach
  for detecting pilot spoofing attack in multiple-antenna systems},''
  \emph{IEEE Trans. Inf. Forensic Secur.}, vol.~10, no.~5, pp. 932--940, May
  2015.

\bibitem{xiong2016secure}
Q.~Xiong, Y.-C. Liang, K.~H. Li, Y.~Gong, and S.~Han, ``{Secure transmission
  against pilot spoofing attack: A two-way training-based scheme},'' \emph{IEEE
  Trans. Inf. Forensic Secur.}, vol.~11, no.~5, pp. 1017--1026, May 2016.

\bibitem{xu2019detection}
S.~Xu, W.~Xu, C.~Pan, and M.~Elkashlan, ``{Detection of jamming attack in
  non-coherent massive SIMO systems},'' \emph{IEEE Trans. Inf. Forensic
  Secur.}, vol.~14, no.~9, pp. 2387--2399, Sep. 2019.

\bibitem{zeng2017enabling}
H.~Zeng, C.~Cao, H.~Li, and Q.~Yan, ``{Enabling jamming-resistant
  communications in wireless MIMO networks},'' in \emph{Proc. IEEE Conf.
  Commun. Netw. Secur. (CNS)}, Las Vegas, NV, United states, Oct. 2017, pp.
  1--9.

\bibitem{gulgun2020massive}
Z.~G{\"u}lg{\"u}n, E.~Bj{\"o}rnson, and E.~G. Larsson, ``{Is massive MIMO
  robust against distributed jammers?}'' \emph{IEEE Trans. Commun.}, vol.~69,
  no.~1, pp. 457--469, Jan. 2020.

\bibitem{akhlaghpasand2019jamming}
H.~Akhlaghpasand, E.~Bj{\"o}rnson, and S.~M. Razavizadeh, ``{Jamming
  suppression in massive MIMO systems},'' \emph{IEEE Trans. Circuits Syst.
  II-Express Briefs}, vol.~67, no.~1, pp. 182--186, Jan. 2020.

\bibitem{do2017jamming}
T.~T. Do, E.~Bj{\"o}rnson, E.~G. Larsson, and S.~M. Razavizadeh,
  ``{Jamming-resistant receivers for the massive MIMO uplink},'' \emph{IEEE
  Trans. Inf. Forensic Secur.}, vol.~13, no.~1, pp. 210--223, Jan. 2018.

\bibitem{akhlaghpasand2020jamming}
H.~Akhlaghpasand, E.~Bj{\"o}rnson, and S.~M. Razavizadeh, ``{Jamming-robust
  uplink transmission for spatially correlated massive MIMO systems},''
  \emph{IEEE Trans. Commun.}, vol.~68, no.~6, pp. 3495--3504, Jun. 2020.

\bibitem{yang2022jamming}
Z.~Yang, H.~Shen, W.~Xu, and C.~Zhao, ``{Jamming suppression for uplink massive
  MIMO systems: A semi-blind receiver design},'' in \emph{Proc. IEEE/CIC Int.
  Conf. Commun. China (ICCC)}, Sanshui, Foshan, China, Aug. 2022, pp. 678--683.

\bibitem{bagherinejad2021direction}
S.~Bagherinejad and S.~M. Razavizadeh, ``{Direction-based jamming detection and
  suppression in mmWave massive MIMO networks},'' \emph{IET Commun.}, vol.~15,
  no.~14, pp. 1780--1790, Apr. 2021.

\bibitem{kim2021adversarial}
B.~Kim, Y.~Sagduyu, T.~Erpek, and S.~Ulukus, ``{Adversarial attacks on deep
  learning based mmWave beam prediction in 5G and beyond},'' in \emph{Proc.
  IEEE Workshop Stat. Signal Process. (SSP)}, Virtual, Rio de Janeiro, Brazil,
  Jul. 2021, pp. 590--594.

\bibitem{darsena2022anti}
D.~Darsena and F.~Verde, ``{Anti-jamming beam alignment in millimeter-wave MIMO
  systems},'' \emph{IEEE Trans. Commun.}, vol.~70, no.~8, pp. 5417--5433, Aug.
  2022.

\bibitem{hou2023music}
Y.~Hou, K.~Yano, N.~Suga, S.~Denno, and T.~Sakano, ``{MUSIC spectrum based
  interference detection and localization for mmWave RIS-MIMO system},'' in
  \emph{Proc. Int. Conf. Adv. Commun. Technol. (ICACT)}, Pyeongchang, Korea,
  Feb. 2023, pp. 1--6.

\bibitem{dinh2023defensive}
S.~Dinh-Van, T.~M. Hoang, B.~B. Cebecioglu, D.~S. Fowler, Y.~K. Mo, and M.~D.
  Higgins, ``{A defensive strategy against beam training attack in 5G mmWave
  networks for manufacturing},'' \emph{IEEE Trans. Inf. Forensic Secur.},
  vol.~18, pp. 2204--2217, Apr. 2023.

\bibitem{li2022spatial}
N.~Li, Y.~Gao, K.~Xu, M.~Guo, N.~Sha, X.~Wang, and G.~Wang, ``{Spatial
  sparsity-based pilot attack detection and transmission countermeasure for
  cell-free massive MIMO system},'' \emph{IEEE Syst. J.}, vol.~17, no.~2, pp.
  2065--2076, Jun. 2023.

\bibitem{zhang2023incremental}
C.~Zhang, L.~Chen, L.~Zhang, Y.~Huang, and W.~Zhang, ``{Incremental
  collaborative beam alignment for millimeter wave cell-free MIMO systems},''
  \emph{IEEE Trans. Commun.}, vol.~71, no.~11, pp. 6377--6390, Nov. 2023.

\bibitem{sanguinetti2019toward}
L.~Sanguinetti, E.~Bj{\"o}rnson, and J.~Hoydis, ``{Toward massive MIMO 2.0:
  Understanding spatial correlation, interference suppression, and pilot
  contamination},'' \emph{IEEE Trans. Commun.}, vol.~68, no.~1, pp. 232--257,
  Jan. 2020.

\bibitem{lim2020efficient}
S.~H. Lim, S.~Kim, B.~Shim, and J.~W. Choi, ``{Efficient beam training and
  sparse channel estimation for millimeter wave communications under
  mobility},'' \emph{IEEE Trans. Commun.}, vol.~68, no.~10, pp. 6583--6596,
  Oct. 2020.

\bibitem{wang2018detection}
X.~Wang, G.~Li, and P.~K. Varshney, ``{Detection of sparse signals in sensor
  networks via locally most powerful tests},'' \emph{IEEE Signal Process.
  Lett.}, vol.~25, no.~9, pp. 1418--1422, Sep. 2018.

\bibitem{mohammadi2022generalized}
A.~Mohammadi, D.~Ciuonzo, A.~Khazaee, and P.~S. Rossi, ``{Generalized locally
  most powerful tests for distributed sparse signal detection},'' \emph{IEEE
  Trans. Signal Inf. Proc. Netw.}, vol.~8, pp. 528--542, Jun. 2022.

\bibitem{loyka2001channel}
S.~L. Loyka, ``{Channel capacity of MIMO architecture using the exponential
  correlation matrix},'' \emph{IEEE Commun. Lett.}, vol.~5, no.~9, pp.
  369--371, Sep. 2001.

\bibitem{zhang2024interleaved}
C.~Zhang, C.~Liu, Y.~Jing, M.~Ding, and Y.~Huang, ``{Interleaved training for
  massive MIMO downlink via exploring spatial correlation},'' \emph{IEEE Trans.
  Wirel. Commun.}, vol.~23, no.~8, pp. 8896--8909, Aug. 2024.

\bibitem{forenza2007simplified}
A.~Forenza, D.~J. Love, and R.~W. Heath, ``{Simplified spatial correlation
  models for clustered MIMO channels with different array configurations},''
  \emph{IEEE Trans. Veh. Technol.}, vol.~56, no.~4, pp. 1924--1934, Jul. 2007.

\bibitem{provost1996exact}
S.~B. Provost and E.~M. Rudiuk, ``{The exact distribution of indefinite
  quadratic forms in noncentral normal vectors},'' \emph{Ann. Inst. Stat.
  Math.}, vol.~48, no.~2, pp. 381--394, Jun. 1996.

\end{thebibliography}


\end{document}